\newcounter{myctr}
\def\myitem{\refstepcounter{myctr}\bibfont\noindent\ifnum\themyctr>9\else\phantom{0}\fi\hangindent17pt\themyctr.\enskip}

\documentclass{ws-ijqi}
\usepackage[colorlinks=true,linkcolor=blue,citecolor=red,plainpages=false,pdfpagelabels]{hyperref}
\usepackage[super,sort,compress]{cite}
\usepackage{mathrsfs,bbold,dsfont}

\newtheorem{protocol}{Protocol}
\newcommand{\1}{\mathbb{1}}
\newcommand{\id}{\mathrm{id}}
\newcommand{\supp}{\mathrm{supp}}
\newcommand{\tr}{\mathrm{Tr}}
\newcommand{\bra}[1]{\langle#1|}
\newcommand{\ket}[1]{|#1\rangle}

\newcommand{\op}[2]{\ket{#1}\!\bra{#2}}

\newcommand{\abb}[1]{{\textnormal{#1}}} 
\newcommand{\bc}[1]{{\widetilde{#1}}} 
\newcommand{\ch}[1]{{\mathcal{#1}}} 
\newcommand{\cl}[1]{{\overline{#1}}} 
\newcommand{\dual}[1]{{\widehat{#1}}} 
\newcommand{\g}[1]{{\widehat{#1}}} 
\newcommand{\s}[1]{{\mathscr{#1}}} 
\newcommand{\sa}[1]{{\mathsf{#1}}} 
\newcommand{\spa}[1]{{\mathds{#1}}} 
\newcommand{\stc}[1]{{\widetilde{#1}}} 

\allowdisplaybreaks

\begin{document}



\title{POSTSELECTED COMMUNICATION OVER QUANTUM CHANNELS}

\author{KAIYUAN JI
}

\address{School of Electrical and Computer Engineering, Cornell University\\
Ithaca, New York 14850, USA
\\
kj264@cornell.edu}

\author{BARTOSZ REGULA}

\address{Mathematical Quantum Information RIKEN Hakubi Research Team, RIKEN Cluster for Pioneering Research (CPR) and RIKEN Center for Quantum Computing (RQC)\\
Wako, Saitama 351-0198, Japan\\
bartosz.regula@gmail.com}

\author{MARK M. WILDE}

\address{School of Electrical and Computer Engineering, Cornell University\\
Ithaca, New York 14850, USA\\
wilde@cornell.edu}

\maketitle



\begin{abstract}
\noindent \textit{Dedicated to Alexander S. Holevo on the occasion of his 80th birthday. Professor Holevo's numerous seminal contributions, going back to Ref.~\citen{holevo_1973-1}, have served as an inspiration for generations of quantum information scientists.} \\

The single-letter characterisation of the entanglement-assisted capacity of a quantum channel is one of the seminal results of quantum information theory. In this paper, we consider a modified communication scenario in which the receiver is allowed an additional, `inconclusive' measurement outcome, and we employ an error metric given by the error probability in decoding the transmitted message conditioned on a conclusive measurement result. We call this setting \emph{postselected communication} and the ensuing highest achievable rates the \emph{postselected capacities}. 
Here, we provide a precise single-letter characterisation of postselected capacities in the setting of entanglement assistance as well as the more general nonsignalling assistance, establishing that they are both equal to the channel's projective mutual information --- a variant of mutual information based on the Hilbert projective metric. We do so by establishing bounds on the one-shot postselected capacities, with a lower bound that makes use of a postselected teleportation-based protocol and an upper bound in terms of the postselected hypothesis testing relative entropy. As such, we obtain fundamental limits on a channel's ability to communicate even when this strong resource of postselection is allowed, implying limitations on communication even when the receiver has access to postselected closed timelike curves.
\end{abstract}


\tableofcontents  

\markboth{Kaiyuan Ji, Bartosz Regula, and Mark M. Wilde}
{Postselected communication over quantum channels}


\section{Introduction}
\label{sec:introduction}

The capacity of a communication channel represents its maximum possible capability to transmit information. 
It is a well-studied and well-understood problem for classical channels~\cite{shannon_1948,cover_2006}, and early seminal contributions by Holevo and others~\cite{holevo_1973-1,Holevo-S-W,H-Schumacher-Westmoreland,Lloyd-S-D,L-Shor-D,L-S-Devetak,bennett_1996,bennett_1999-1,bennett_2002,holevo_2002} suggested that many parallels between the classical and quantum settings may hold. However, the subsequent discovery of phenomena which do not exist in the classical realm, such as superadditivity~\cite{shor_2005,hastings_2009} and superactivation~\cite{smith_2008,smith_2011}, meant that the characterisation of channel capacities in the quantum setting is significantly more complicated.
Luckily, it was found that allowing the communicating parties access to some shared quantum resources --- for instance, shared entanglement --- can significantly simplify the characterisation of the communication capacities~\cite{bennett_2002,bennett_2014,berta_2011}. The authors of Ref.~\citen{bennett_2002} argued that such an entanglement-assisted setting is a natural generalisation of the capacity of a classical channel to the quantum case, as it leads to closer similarities between the two. Furthermore, a quantum feedback channel does not enhance this capacity~\cite{bowen_2004}, in analogy with the fact that feedback does not enhance the capacity in the classical setting~\cite{shannon_1956}. Various other assisted communication settings have also been studied~\cite{bennett_1996,leung_2015,duan_2016} (see Refs.~\citen{hayashi_2017-book,wilde_2017,watrous_2018,holevo_2019,khatri_2024} for textbooks on the topic).

In addition to simplifying the computation of asymptotic communication capacities, such approaches have the additional benefit that they allow us to understand the ultimate limits of communication through quantum channels: even with assisting quantum resources at our disposal, the bounds given by the assisted capacities can never be exceeded. This motivates the understanding of the precise advantages that can be gained through access to specific resources, and in particular the extent of the advantages that can be enabled by assisted communication schemes.

To gain insight into the limitations of communication over quantum channels in broader settings,
we draw inspiration from the recently introduced framework of postselected quantum hypothesis testing~\cite{regula_2022-4}, in which the usual setting of quantum state discrimination is extended by allowing an additional, `inconclusive' measurement outcome --- representing, for example, situations in which a given process does not conclusively distinguish between the states in consideration, resulting in no guess being made. Although such inconclusive approaches have been explored in both classical~\cite{forney_1968,merhav_2008,tan_2014,hayashi_2015} and quantum information~\cite{ivanovic_1987,dieks_1988,peres_1988,chefles_1998, fiurasek_2003, rudolph_2003, croke_2006, herzog_2005, herzog_2009} before, the crucial difference in the framework of Ref.~\citen{regula_2022-4} is that error rates are evaluated only after conditioning on a conclusive result. This is shown to lead to significant and wide-ranging simplifications, making many one-shot and asymptotic quantities almost effortlessly computable.
The price to pay in such an approach is the need to allow the communicating parties arbitrary access to postselection. Such a concession is known to significantly enhance the power of quantum mechanics in many contexts~\cite{aaronson_2005,fiurasek_2006,gendra_2012,combes_2014,combes_2015,arvidsson-shukur_2020,gisin_1996,kent_1998,horodecki_1999-1,reeb_2011,regula_2022,regula_2022-2}, even being equivalent to having access to postselected closed timelike curves~\cite{lloyd_2011,lloyd_2011-1} --- a well-known extension of  quantum mechanics that has been previously shown to increase its information-processing and computational power \cite{brun_2012}.
However, because of the remarkable simplifications that follow, postselection can be very useful in understanding the ultimate power of quantum information processing, and it has already found use as the conceptual foundation of the hardness arguments that underlie quantum supremacy experiments~\cite{harrow_2017}.

Here we study classical and quantum communication over quantum channels in a setting where the receiver is allowed to perform inconclusive decoding schemes on the message transmitted through the channel. The probability of error, as well as the resulting capacity, is then evaluated conditioned on a conclusive outcome.  Here we compute the quantum and classical capacities in two scenarios: postselected communication with entanglement assistance (pEA), as well as postselected communication assisted by nonsignalling correlations (pNA).  We show, in particular, that the two settings are essentially equivalent --- not only asymptotically, but also at the one-shot level --- and the resulting capacity can be evaluated exactly as a single-letter quantity. Specifically, we show for every channel $\ch{N}_{A\to B}$ that
\begin{align}
\label{eq:main-result-intro}
	C_\sa{pEA}(\ch{N})&=C_\sa{pNA}(\ch{N})=2Q_\sa{pEA}(\ch{N})=2Q_\sa{pNA}(\ch{N})=I_\Omega(\ch{N}),
\end{align}
where $C$ denotes classical capacity, $Q$ quantum capacity, and $I_\Omega(\ch{N})$ is a variant of the mutual information of a channel based on the Hilbert projective metric, which can be efficiently computed as a semidefinite program (see Eq.~\eqref{eq:projective-mutual-channel}).  Our achievability results are based on a postselected communication scheme inspired by probabilistic quantum teleportation.

Our finding in Eq.~\eqref{eq:main-result-intro} provides a complete solution to the problem of  postselected communication over quantum channels with entanglement or nonsignalling assistance, shedding light on the ultimate limits of quantum channels to transmit information even in permissive settings, and once again showing that the addition of postselection leads to a major simplification of the computation of asymptotic rates of quantum information-processing tasks.

The rest of our paper proceeds as follows. In Sec.~\ref{sec:notation}, we establish notation and basic definitions used throughout. Sec.~\ref{sec:communication} provides precise definitions of postselected communication and associated capacities, beginning with a general framework (Sec.~\ref{sec:framework}), and followed by the more specific settings of postselected entanglement-assisted communication (Sec.~\ref{sec:pEA}) and postselected nonsignalling-assisted communication (Sec.~\ref{sec:pNA}). Our main results are presented in Secs.~\ref{sec:oneshot} and \ref{sec:asymptotic}, which give bounds on the one-shot capacities and  precise formulas for the asymptotic capacities, respectively. We finally conclude in Sec.~\ref{sec:conclusion} with a summary of our findings and some directions for future work. Appendices~\ref{app:equivalent}--\ref{app:achievability} give detailed proofs for some claims used in the main text.

\section{Notation}
\label{sec:notation}

Let $\spa{H}_A$ denote the (finite-dimensional) Hilbert space associated with a quantum system $A$.  Let $d_A$ denote the dimensionality of $\spa{H}_A$, and let $[d_A]$ denote the set of consecutive integers $\{1,2,\dots,d_A\}$.  Let $\spa{B}_A$ denote the space of bounded operators acting on $\spa{H}_A$.  A quantum state in $A$ is a positive semidefinite operator $\rho_A\in\spa{B}_A$ of unit trace, and it is said to be pure whenever it is rank one.  A quantum channel from $A$ to $B$ is a linear map $\ch{N}_{A\to B}\colon\spa{B}_A\to\spa{B}_B$ that is completely positive and trace preserving (CPTP).  The Choi state of a channel $\ch{N}_{A\to B}$, denoted by $\Phi_{RB}^\ch{N}\in\spa{B}_R\otimes\spa{B}_B$, is defined as
\begin{align}
	\Phi_{RB}^\ch{N}&:=\ch{N}_{A\to B}\!\left[\Phi_{RA}\right],
\end{align}
where $R$ is a reference system with $d_R=d_A$ and $\Phi_{RA}\equiv\frac{1}{d_A}\sum_{i,j\in[d_A]}\op{i}{j}_R\otimes\op{i}{j}_A$ denotes the standard maximally entangled state between $R$ and $A$.  The replacement channel $\ch{R}_{A\to B}^\sigma$ parametrised by a state $\sigma_B$ is defined as $\ch{R}_{A\to B}^\sigma\colon\rho_A\mapsto\tr[\rho_A]\sigma_B$, and it satisfies $\Phi_{RB}^{\ch{R}^\sigma}=\frac{1}{d_A}\1_R\otimes\sigma_B$.  Let $\s{R}$ denote the set of replacement channels.

Let us define several information-theoretic measures relevant to this work.  The \emph{max-relative entropy} between two states $\rho_A$ and $\sigma_A$ is defined as~\cite{datta_2009}
\begin{align}
\label{eq:max-relative}
	D_{\max}(\rho\|\sigma)&:=\log_2\inf_{\lambda\in\spa{R}}\left\{\lambda\colon\rho_A\leq\lambda\sigma_A\right\} \notag\\
	&\hphantom{:}=\log_2\sup_{O_A}\left\{\frac{\tr\!\left[O_A\rho_A\right]}{\tr\!\left[O_A\sigma_A\right]}\colon0\leq O_A\leq\1_A\right\},
\end{align}
where the infimum in the first line is over all real numbers $\lambda$ and the second line is known as the dual formulation~\cite[Lemma~A.4]{mosonyi_2013}.  The \emph{max-relative entropy} between two channels $\ch{N}_{A\to B}$ and $\ch{M}_{A\to B}$ is defined as~\cite{leditzky_2018-2}
\begin{align}
\label{eq:max-relative-channel}
	D_{\max}(\ch{N}\|\ch{M})&:=\sup_{\rho_{RA}}D_{\max}(\ch{N}_{A\to B}\!\left[\rho_{RA}\right]\|\ch{M}_{A\to B}\!\left[\rho_{RA}\right]),
\end{align}
where the supremum is over all systems $R$ and states $\rho_{RA}$.  The supremum in Eq.~\eqref{eq:max-relative-channel} is known to be achieved by the maximally entangled state $\Phi_{RA}$ with $d_R=d_A$~\cite{diaz_2018-2, wilde_2020}, which implies that
\begin{align}
\label{eq:max-relative-choi}
	D_{\max}(\ch{N}\|\ch{M})=D_{\max}(\Phi^\ch{N}\|\Phi^\ch{M}).
\end{align}
The \emph{Hilbert projective metric} between two states $\rho_A$ and $\sigma_A$ is defined as~\cite{bushell_1973, reeb_2011}
\begin{align}
\label{eq:projective-relative}
	D_\Omega(\rho\|\sigma)&:=D_{\max}(\rho\|\sigma)+D_{\max}(\sigma\|\rho).
\end{align}
Due to the fact that $D_{\max}(\rho\|\sigma)=\infty$ if and only if $\supp(\rho)\not\subseteq\supp(\sigma)$, we note that $D_\Omega(\rho\|\sigma)=\infty$ if and only if $\supp(\rho)\neq\supp(\sigma)$.

Analogously to how the quantum mutual information can be defined in terms of the quantum relative entropy (see, e.g., Ref.~\cite[Exercise~11.8.2]{wilde_2017}), we define the \emph{projective mutual information} between the systems $A$ and $B$ of a bipartite state $\rho_{AB}$, based on the Hilbert projective metric, as
\begin{align}
\label{eq:projective-mutual}
	I_\Omega(A;B)_{\rho_{AB}}&:=\inf_{\sigma_B}D_\Omega(\rho_{AB}\|\rho_A\otimes\sigma_B),
\end{align}
where $\rho_A\equiv\tr_B[\rho_{AB}]$ and the infimum is over all states $\sigma_B$.  We also define the \emph{projective mutual information} of a channel $\ch{N}_{A\to B}$ as
\begin{align}
\label{eq:projective-mutual-channel}
	I_\Omega(\ch{N})&:=\sup_{\rho_{RA}}I_\Omega(R;B)_{\ch{N}_{A\to B}\!\left[\rho_{RA}\right]},
\end{align}
where the supremum is over all systems $R$ and states $\rho_{RA}$.  Importantly, the projective mutual information of $\ch{N}_{A\to B}$ can be expressed in terms of a simple semidefinite program:
\begin{align}
\label{eq:sdp_main}
	I_\Omega(\ch{N})&=\log_2\inf_{\substack{\xi\in\spa{R}, \\ S_B\geq0}}\left\{\xi\colon\Phi_{RB}^{\ch{N}}\leq\1_R\otimes S_B\leq\xi\Phi_{RB}^{\ch{N}}\right\}.
\end{align}
We prove this formulation, along with other equivalent expressions for $I_\Omega(\ch{N})$, in Appendix~\ref{app:equivalent}.

Postselected hypothesis testing, recently proposed in Ref.~\citen{regula_2022-4}, is a fundamental information-processing task in the setting where postselection is allowed.  As in conventional hypothesis testing, the experimenter is given one of two possible states, $\rho$ (null hypothesis) or $\sigma$ (alternative hypothesis), and is asked to guess which state they actually receive --- but now they are allowed an additional option of making no guess without being penalised.  In this scenario, the experimenter's strategy can in general be described by a ternary positive operator-valued measure (POVM) $\{P,Q,\1-P-Q\}$, such that the first two outcomes correspond to guessing $\rho$ and guessing $\sigma$, respectively, and the last outcome corresponds to making no guess, which we also refer to as being `inconclusive.'  Since the experimenter is not penalised for being inconclusive, the error probabilities of interest are those conditioned on a conclusive outcome.  Specifically, the \emph{conditional type-I error} probability, of mistaking $\rho$ as $\sigma$, is given by
\begin{align}
	\alpha_\abb{pH}(\rho;P,Q)&:=\frac{\tr\!\left[Q\rho\right]}{\tr\!\left[\left(P+Q\right)\rho\right]},
\end{align}
and the \emph{conditional type-II error} probability, of mistaking $\sigma$ as $\rho$, is given by
\begin{align}
	\beta_\abb{pH}(\sigma;P,Q)&:=\frac{\tr\!\left[P\sigma\right]}{\tr\!\left[\left(P+Q\right)\sigma\right]}.
\end{align}
It is then natural to define the \emph{postselected $\varepsilon$-hypothesis testing relative entropy} between $\rho$ and $\sigma$~\cite{regula_2022-4}, as follows:
\begin{align}
\label{eq:pH-relative}
	D_\abb{pH}^\varepsilon(\rho\|\sigma)&:=-\log_2\inf_{P,Q\geq0}\left\{\beta_\abb{pH}(\sigma;P,Q)\colon\alpha_\abb{pH}(\rho;P,Q)\leq\varepsilon,\;P+Q\leq\1\right\} \notag\\
	&\hphantom{:}=-\log_2\inf_{P,Q\geq0}\left\{\frac{\tr\!\left[P\sigma\right]}{\tr\!\left[\left(P+Q\right)\sigma\right]}\colon\frac{\tr\!\left[Q\rho\right]}{\tr\!\left[\left(P+Q\right)\rho\right]}\leq\varepsilon,\;P+Q\leq\1\right\},
\end{align}
where the supremum is over all positive semidefinite operators $P$ and $Q$.  In contrast to the conventional setting, the postselected hypothesis testing relative entropy has a closed-form expression in terms of the Hilbert projective metric, given by the following lemma.

\begin{lemma}[Characterisation of $D_\abb{pH}^\varepsilon$ {\cite[Theorem~1]{regula_2022-4}}]
\label{lem:pH-projective}
Let $\rho_A$ and $\sigma_A$ be two states, and let $\varepsilon\in(0,1)$.  Then
\begin{align}
	D_\abb{pH}^\varepsilon(\rho\|\sigma)&=\log_2\!\left(\frac{\varepsilon}{1-\varepsilon}2^{D_\Omega(\rho\|\sigma)}+1\right).
\end{align}
\end{lemma}

Since the Hilbert projective metric satisfies the data-processing inequality under positive linear maps~\cite[Proposition~3]{regula_2022-4}, Lemma~\ref{lem:pH-projective} implies that the postselected hypothesis testing relative entropy has the same property:
\begin{align}
\label{eq:data-processing}
	D_\abb{pH}^\varepsilon(\rho\|\sigma)&\geq D_\abb{pH}^\varepsilon(\ch{N}\!\left[\rho\right]\|\ch{N}\!\left[\sigma\right])
\end{align}
for all positive linear maps $\ch{N}_{A\to B}$ and states $\rho_A$ and $\sigma_A$.

\section{Postselected communication}
\label{sec:communication}

In this section, we lay out the general framework of postselected communication over quantum channels and introduce two specific communication scenarios: postselected entanglement-assisted communication and postselected nonsignalling-assisted communication.

\subsection{General framework}
\label{sec:framework}

Inspired by postselected hypothesis testing, we propose the task of \emph{postselected communication} over quantum channels.  Compared to conventional communication settings, the key difference here is that the receiver's decoding operation is allowed an additional option of being inconclusive about the message being transmitted, without being penalised.  Consequently, the figure of merit for message transmission is the accuracy conditioned on when the receiver \emph{conclusively} decodes the message.  To better demonstrate the setting, let us consider the following example, in which a classical message $m\in[d_M]$ is to be transmitted.  Let $\g{M}\in[d_M]\cup\{\perp\}$ be a random variable corresponding to the decoder's recovery of $m$, where $\perp$ denotes the inconclusive outcome, so that
\begin{align}
	\sum_{m'\in[d_M]}\Pr_\g{M}\left\{\g{M}=m'\right\}+\Pr_\g{M}\left\{\g{M}=\,\perp\right\}&=1.
\end{align}
The conditional error probability of the transmission of $m$ is then given by
\begin{align}
	\Pr_\g{M}\left\{\g{M}\neq m\,\middle|\,\abb{conclusive}\right\}&=\Pr_\g{M}\left\{\g{M}\neq m\,\middle|\,\g{M}\neq\,\perp\right\} \notag\\
	&=\frac{\sum_{m'\in[d_M]\colon m'\neq m}\Pr_\g{M}\left\{\g{M}=m'\right\}}{\sum_{m'\in[d_M]}\Pr_\g{M}\left\{\g{M}=m'\right\}} \notag\\
	&=1-\frac{\Pr_\g{M}\left\{\g{M}=m\right\}}{\sum_{m'\in[d_M]}\Pr_\g{M}\left\{\g{M}=m'\right\}}. \label{eq:error-informal}
\end{align}

Now let us put the discussion in formal terms, and furthermore, allow for the message more generally to be classical (as above) or quantum.  Consider a given quantum channel $\ch{N}_{A\to B}$ connecting Alice (the sender) and Bob (the receiver).  Let $M$ denote the message source system on Alice's side, which may be quantum or classical.  Let $\g{M}$ denote the system storing Bob's recovery of the message, which satisfies $d_\g{M}=d_M$.  Apart from $\g{M}$, Bob also holds a classical flag system $X$ satisfying $d_X=2$ and indicating whether his decoding is conclusive ($X=1$) or not ($X=0$)~\footnote{Note that compared to the previous example, we have kept the information about Bob's (in)conclusiveness in a separate system from the recovered message --- this is purely for notational convenience, and both formulations are equivalent.}.  Considering that Bob will postselect the conclusive outcome in $X$, the most general scheme for Alice and Bob to conduct one-shot postselected communication from $M$ to $\g{M}$ through the channel $\ch{N}_{A\to B}$ is in effect equivalent to using $\ch{N}_{A\to B}$ to simulate a `subnormalised' channel $\ch{N}_{M\to\g{M}}'\colon\spa{B}_M\to\spa{B}_\g{M}$ via a probabilistic channel transformation $\Theta_{(A\to B)\to(M\to\g{M})}$~\cite{chiribella_2008, burniston_2020}, as illustrated in Fig.~\ref{fig:postselected}:
\begin{align}
\label{eq:transformation}
	\ch{N}_{M\to\g{M}}'&\equiv\Theta_{(A\to B)\to(M\to\g{M})}\left\{\ch{N}_{A\to B}\right\},
\end{align}
Here $\ch{N}_{M\to\g{M}}'$ being a subnormalised channel means that it is completely positive and trace nonincreasing, i.e., for every system $R$ and state $\rho_{RM}$,
\begin{align}
	\ch{N}_{M\to\g{M}}'\!\left[\rho_{RM}\right]&\geq0, \\
	\tr\!\left[\ch{N}_{M\to\g{M}}'\!\left[\rho_{RM}\right]\right]&\leq1.
\end{align}
Accommodating the terminologies in the literature, we will henceforth refer to subnormalised channels as `subchannels' and probabilistic channel transformations as `probabilistic supermaps.'  

\begin{figure}[t]
\centerline{\includegraphics[scale=0.24]{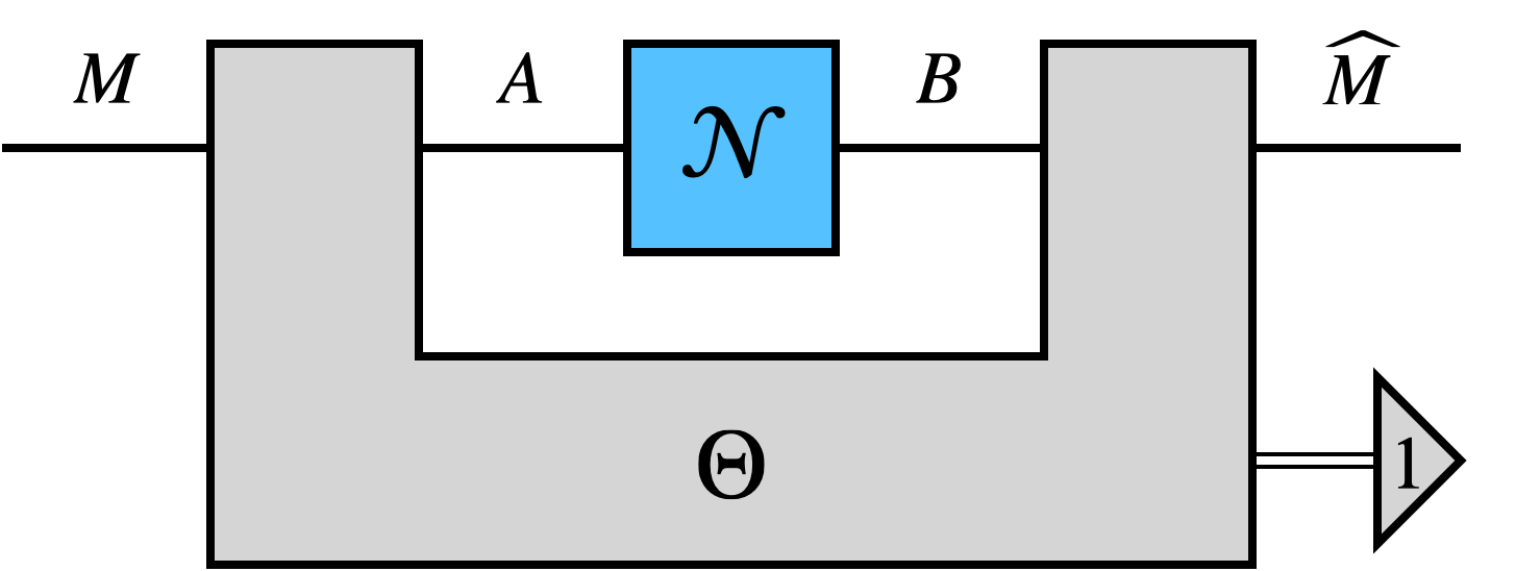}}
\vspace*{8pt}
\caption{Simulation of a subchannel $\ch{N}_{M\to\g{M}}'$ with a probabilistic supermap $\Theta_{(A\to B)\to(M\to\g{M})}$ over a channel $\ch{N}_{A\to B}$ [see Eq.~\eqref{eq:transformation}].  Note that $\Theta_{(A\to B)\to(M\to\g{M})}$ includes postselecting the conclusive outcome in the classical flag system.}
\label{fig:postselected}
\end{figure}

Following the intuition of Eq.~\eqref{eq:error-informal}, given the channel $\ch{N}_{A\to B}$ and the probabilistic supermap $\Theta_{(A\to B)\to(M\to\g{M})}$, when the system $M$ is classical (and so is $\g{M}$), the worst-case conditional error probability of classical communication is defined as
\begin{align}
\label{eq:error-classical}
	P_\abb{err}^\abb{(c)}(\Theta;\ch{N})&:=1-\min_{m\in[d_M]}\frac{\tr\!\left[\op{m}{m}_\g{M}\ch{N}_{M\to\g{M}}'\!\left[\op{m}{m}_M\right]\right]}{\tr\!\left[\ch{N}_{M\to\g{M}}'\!\left[\op{m}{m}_M\right]\right]},
\end{align}
where $\ch{N}_{M\to\g{M}}'$ follows Eq.~\eqref{eq:transformation}.  Likewise, when the system $M$ is quantum (and so is $\g{M}$), the worst-case conditional error probability of quantum communication is defined as
\begin{align}
\label{eq:error-quantum}
	P_\abb{err}^\abb{(q)}(\Theta;\ch{N})&:=1-\inf_{\psi_{RM}}\frac{\tr\!\left[\psi_{R\g{M}}\ch{N}_{M\to\g{M}}'\!\left[\psi_{RM}\right]\right]}{\tr\!\left[\ch{N}_{M\to\g{M}}'\!\left[\psi_{RM}\right]\right]},
\end{align}
where the infimum is over all systems $R$ and pure states $\psi_{RM}$.  The following definition captures the efficiency of a probabilistic supermap as a communication protocol, in terms of the size of the message system and the conditional error probability.

\begin{definition}[$(d_M,\varepsilon)$ protocols]
\label{def:protocol}
Let $\ch{N}_{A\to B}$ be a quantum channel, and let $\varepsilon\in(0,1)$.  A probabilistic supermap $\Theta_{(A\to B)\to(M\to\g{M})}$ is called a $(d_M,\varepsilon)$ classical protocol over $\ch{N}_{A\to B}$ whenever $P_\abb{err}^\abb{(c)}(\Theta;\ch{N})\leq\varepsilon$, and it is called a $(d_M,\varepsilon)$ quantum protocol over $\ch{N}_{A\to B}$ whenever $P_\abb{err}^\abb{(q)}(\Theta;\ch{N})\leq\varepsilon$.
\end{definition}

When Alice and Bob have access to a particular type of resource, the encoding and decoding operations are restricted to a particular class of probabilistic supermaps.  Given a class $\s{T}$ of available probabilistic supermaps, the one-shot postselected classical or quantum capacity of the channel $\ch{N}_{A\to B}$ can be defined as the maximum number of classical or quantum bits transmitted, optimised over protocols within $\s{T}$, subject to a given conditional error probability.

\begin{definition}[One-shot $\s{T}$-assisted capacities]
\label{def:oneshot-capacities}
Let $\s{T}$ be a class of probabilistic supermaps, and let $\varepsilon\in[0,1]$.  The one-shot $\varepsilon$-error $\s{T}$-assisted classical capacity of a quantum channel $\ch{N}_{A\to B}$ is defined as
\begin{align}
	\label{eq:one-shot-classical-capacity}
	C_\s{T}^\varepsilon(\ch{N})&:=\sup_{\Theta\in\s{T}}\left\{\log_2 d_M\colon P_\abb{err}^\abb{(c)}(\Theta;\ch{N})\leq\varepsilon\right\}.
\end{align}
The one-shot $\varepsilon$-error $\s{T}$-assisted quantum capacity of $\ch{N}_{A\to B}$ is defined as
\begin{align}
	\label{eq:oneshot-quantum-capacity}
	Q_\s{T}^\varepsilon(\ch{N})&:=\sup_{\Theta\in\s{T}}\left\{\log_2 d_M\colon P_\abb{err}^\abb{(q)}(\Theta;\ch{N})\leq\varepsilon\right\}.
\end{align}
\end{definition}

When Alice and Bob are connected by multiple instances of the same channel in parallel, the asymptotic rate of transmitted bits per channel use is known as the asymptotic capacity of the channel.

\begin{definition}[Asymptotic $\s{T}$-assisted capacities]
\label{def:asymptotic-capacities}
Let $\s{T}$ be a class of probabilistic supermaps.  The asymptotic $\s{T}$-assisted classical capacity of a quantum channel $\ch{N}_{A\to B}$ is defined as
\begin{align}
\label{eq:asymptotic-pEA-classical}
	C_\s{T}(\ch{N})&:=\inf_{\varepsilon\in(0,1)}\liminf_{n\to\infty}\frac{1}{n}C_\s{T}^\varepsilon(\ch{N}^{\otimes n}).
\end{align}
The asymptotic $\s{T}$-assisted quantum capacity of $\ch{N}_{A\to B}$ is defined as
\begin{align}
\label{eq:asymptotic-quantum-capacity}
	Q_\s{T}(\ch{N})&:=\inf_{\varepsilon\in(0,1)}\liminf_{n\to\infty}\frac{1}{n}Q_\s{T}^\varepsilon(\ch{N}^{\otimes n}).
\end{align}
The strong converse $\s{T}$-assisted classical capacity of $\ch{N}_{A\to B}$ is defined as
\begin{align}
\label{eq:strong-pEA-classical}
	\stc{C}_\s{T}(\ch{N})&:=\sup_{\varepsilon\in(0,1)}\limsup_{n\to\infty}\frac{1}{n}C_\s{T}^\varepsilon(\ch{N}^{\otimes n}).
\end{align}
The strong converse $\s{T}$-assisted quantum capacity of $\ch{N}_{A\to B}$ is defined as
\begin{align}
\label{eq:strong-quantum-capacity}
	\stc{Q}_\s{T}(\ch{N})&:=\sup_{\varepsilon\in(0,1)}\limsup_{n\to\infty}\frac{1}{n}Q_\s{T}^\varepsilon(\ch{N}^{\otimes n}).
\end{align}
\end{definition}

\begin{remark}[Note on related work]
Let us note that the postselected communication setting is related to, yet distinct from, the well-known setting of decoding with erasure, as considered in the classical information theory literature~\cite{forney_1968,merhav_2008,tan_2014,hayashi_2015} as well as in quantum information~\cite{filippov_2021,filippov_2022}. Although these frameworks allow for inconclusive measurement results, their error metrics are not conditioned on a conclusive outcome, and in fact they require the conclusive probability to tend to one in their definition of capacity.
\end{remark}

The above framework is general enough to capture any postselected communication scenario where Alice and Bob are assisted by a particular type of resource.  In the rest of the paper, we will focus on two such scenarios.  In the first scenario, the class $\s{T}$ consists of probabilistic supermaps realisable with shared entanglement, and in the second scenario, it consists of those realisable with nonsignalling correlations.

\subsection{Postselected entanglement-assisted communication}
\label{sec:pEA}

Entanglement-assisted (EA) communication is a prominent and insightful scenario to investigate when studying the fundamental limits of communication over quantum channels~\cite{bennett_1999-1, bennett_2002}.  In conventional EA communication, Alice and Bob are allowed to perform arbitrary encoding and decoding operations locally, with the assistance of an arbitrary shared entangled state.  In \emph{postselected entanglement-assisted (pEA) communication}, we allow Bob to perform postselection as a part of his decoding operation, on top of a conventional EA protocol.  Specifically, Alice and Bob are allowed to transmit messages from $M$ to $\g{M}$ over a given quantum channel $\ch{N}_{A\to B}$ according to the following procedure, as illustrated in Fig.~\ref{fig:pEA}:
\begin{enumerate}
	\item Alice and Bob share a bipartite state $\gamma_{A'B'}$;
	\item Alice encodes the message in $M$ (classical or quantum) by applying a channel $\ch{E}_{MA'\to A}$;
	\item The system $A$ is passed from Alice to Bob through the channel $\ch{N}_{A\to B}$;
	\item Bob receives the system $B$, decodes the message, and stores it in $\g{M}$ by applying a \emph{subchannel} $\ch{D}_{BB'\to\g{M}}$.
\end{enumerate}
We summarise the composition of a pEA protocol formally as below.

\begin{figure}[t]
\centerline{\includegraphics[scale=0.24]{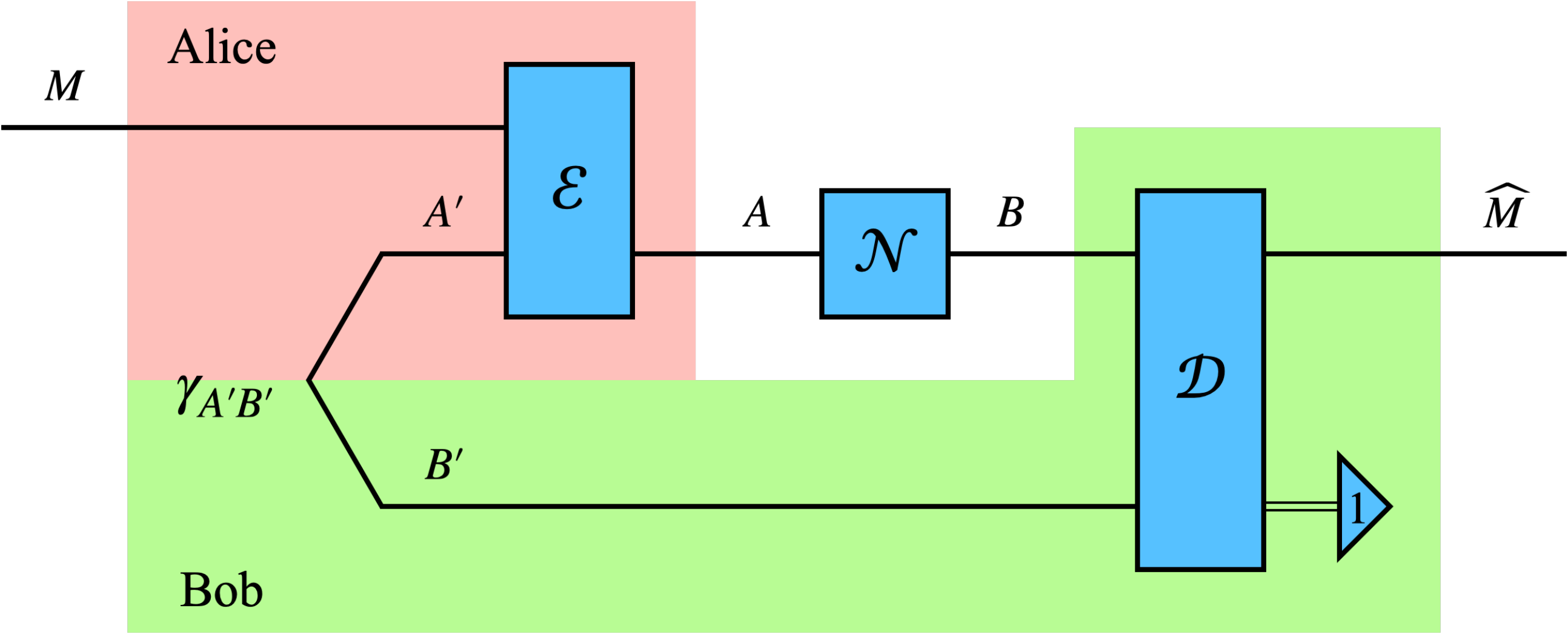}}
\vspace*{8pt}
\caption{Postselected entanglement-assisted communication over a channel $\ch{N}_{A\to B}$ using a protocol $(\gamma_{A'B'},\ch{E}_{MA'\to A},\ch{D}_{BB'\to\g{M}})$ (see Definition~\ref{def:pEA}).  The red and green regions represent Alice's and Bob's operations, respectively.  Note that the decoding operation $\ch{D}_{BB'\to\g{M}}$ includes postselecting the conclusive outcome in the classical flag system.}
\label{fig:pEA}
\end{figure}

\begin{definition}[pEA protocols]
\label{def:pEA}
A probabilistic supermap $\Theta_{(A\to B)\to(M\to\g{M})}\colon$ $\ch{N}_{A\to B}\mapsto\ch{N}_{M\to\g{M}}'$ is called a postselected entanglement-assisted protocol whenever there exists a state $\gamma_{A'B'}$, a channel $\ch{E}_{MA'\to A}$, and a subchannel $\ch{D}_{BB'\to\g{M}}$ such that
\begin{align}
\label{eq:pEA}
	\ch{N}_{M\to\g{M}}'\!\left[\rho_M\right]&=\left(\ch{D}_{BB'\to\g{M}}\circ\ch{N}_{A\to B}\circ\ch{E}_{MA'\to A}\right)\!\left[\rho_M\otimes\gamma_{A'B'}\right]\qquad\forall\rho_M.
\end{align}
The class of pEA protocols is denoted by $\sa{pEA}$.
\end{definition}

It is common that we simply represent the pEA protocol in Definition~\ref{def:pEA} by the triple of its elements: $\Theta_{(A\to B)\to(M\to\g{M})}\equiv(\gamma_{A'B'},\ch{E}_{MA'\to A},\ch{D}_{BB'\to\g{M}})$.

\subsection{Postselected nonsignalling-assisted communication}
\label{sec:pNA}

Nonsignalling-assisted (NA) communication is another theoretically intriguing scenario of assisted communication~\cite{duan_2016, fang_2020-1, takagi_2020}.  In conventional NA communication, Alice and Bob are allowed to implement an arbitrary protocol which itself does not signal from Alice to Bob.  Specifically, every (deterministic) superchannel $\Theta_{(A\to B)\to(M\to\g{M})}$ can be thought of as a bipartite channel $\bc{\Theta}_{MB\to A\g{M}}$ shared between Alice (holding $M$ and $A$) and Bob (holding $B$ and $\g{M}$) (see Ref.~\cite[Fig.~2]{gour_2019-2} under the notation $\Gamma_\Theta$).  The superchannel $\Theta_{(A\to B)\to(M\to\g{M})}$ is called an NA protocol whenever $\bc{\Theta}_{MB\to A\g{M}}$ satisfies the Alice-to-Bob-nonsignalling constraint (also see Ref.~\cite[Eq.~(5)]{takagi_2020}): for all states $\rho_M$, $\omega_M$, and $\sigma_B$,
\begin{align}
\label{eq:NA}
	\tr_A\!\left[\bc{\Theta}_{MB\to A\g{M}}\!\left[\rho_M\otimes\sigma_B\right]\right]&=\tr_A\!\left[\bc{\Theta}_{MB\to A\g{M}}\!\left[\omega_M\otimes\sigma_B\right]\right].
\end{align}
A crucial property of NA protocols is that a superchannel is an NA protocol if and only if it is replacement preserving~\cite[Proposition~1]{takagi_2020}, i.e., always mapping a replacement channel to another replacement channel.  From a resource-theoretic perspective~\cite{chitambar_2019, takagi_2020}, as replacement channels are useless for transmitting information, the class of NA protocols is exactly the maximal set of superchannels that cannot generate useful communication resources from useless ones.  In \emph{postselected nonsignalling-assisted (pNA) communication}, we allow Bob to perform postselection on top of a conventional NA protocol.  This amounts to defining a pNA protocol as follows.

\begin{definition}[pNA protocols]
\label{def:pNA}
A probabilistic supermap $\Theta_{(A\to B)\to(M\to\g{M})}\colon$ $\ch{N}_{A\to B}\mapsto\ch{N}_{M\to\g{M}}'$ is called a postselected nonsignalling assisted protocol whenever there exists an NA protocol $\Xi_{(A\to B)\to(M\to\g{M}X)}$ and a subchannel $\ch{D}_{\g{M}X\to\g{M}}$ such that
\begin{align}
	\ch{N}_{M\to\g{M}}'\!\left[\rho_M\right]&=\left(\ch{D}_{\g{M}X\to\g{M}}\circ\Xi_{(A\to B)\to(M\to\g{M}X)}\left\{\ch{N}_{A\to B}\right\}\right)\!\left[\rho_M\right]\qquad\forall\rho_M.
\end{align}
The class of pNA protocols is denoted by $\sa{pNA}$.
\end{definition}

Let us note here the relation $\sa{pEA}\subset\sa{pNA}$ between the two classes of postselected protocols.  The following proposition generalises Ref.~\cite[Proposition~1]{takagi_2020} by showing that pNA protocols are essentially those that are probabilistically replacement preserving.

\begin{proposition}[Characterisation of pNA protocols]
\label{prop:pNA}
Let $\Theta_{(A\to B)\to(M\to\g{M})}$ be a probabilistic supermap, which can be thought of as a bipartite subchannel $\bc{\Theta}_{MB\to A\g{M}}$ between Alice and Bob (see Fig.~\ref{fig:bipartite}).  Then the following statements are equivalent.
\begin{itemize}
	\item[\textnormal{(i)}] There exists a pNA protocol $\Lambda_{(A\to B)\to(M\to\g{M})}$ and $c\geq0$ such that 
    \begin{align}
        \Theta_{(A\to B)\to(M\to\g{M})}&=c\Lambda_{(A\to B)\to(M\to\g{M})}.
    \end{align}
	\item[\textnormal{(ii)}] $\bc{\Theta}_{MB\to A\g{M}}$ is Alice-to-Bob nonsignalling in the sense of Eq.~\eqref{eq:NA} (see Fig.~\ref{fig:pNA}).
	\item[\textnormal{(iii)}] For every replacement channel $\ch{R}_{A\to B}^\sigma$, there exists a state $\sigma_\g{M}'$ and $p\in[0,1]$ such that
    \begin{align}
        \Theta_{(A\to B)\to(M\to\g{M})}\left\{\ch{R}_{A\to B}^\sigma\right\}&=p\ch{R}_{M\to\g{M}}^{\sigma'}.
    \end{align}
\end{itemize}
\end{proposition}

\begin{figure}[t]
\centerline{\includegraphics[scale=0.2]{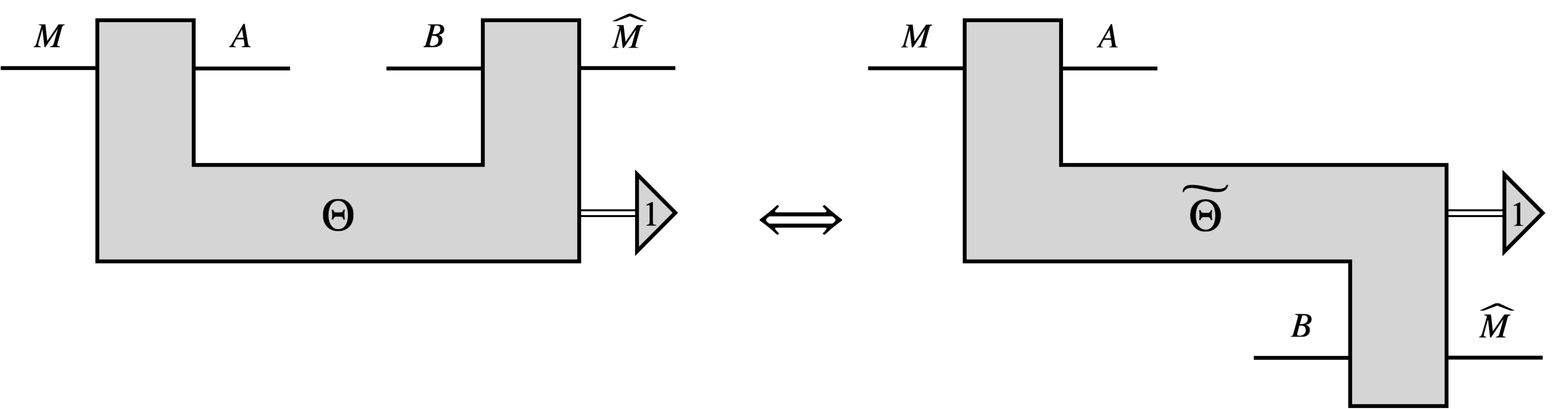}}
\vspace*{8pt}
\caption{Equivalence between a probabilistic supermap $\Theta_{(A\to B)\to(M\to\g{M})}$ (left) and the bipartite subchannel $\bc{\Theta}_{MB\to A\g{M}}$ (right).  The systems $M$ and $A$ are held by Alice, and the systems $B$ and $\g{M}$ are held by Bob.}
\label{fig:bipartite}
\end{figure}

\begin{figure}[t]
\centerline{\includegraphics[scale=0.2]{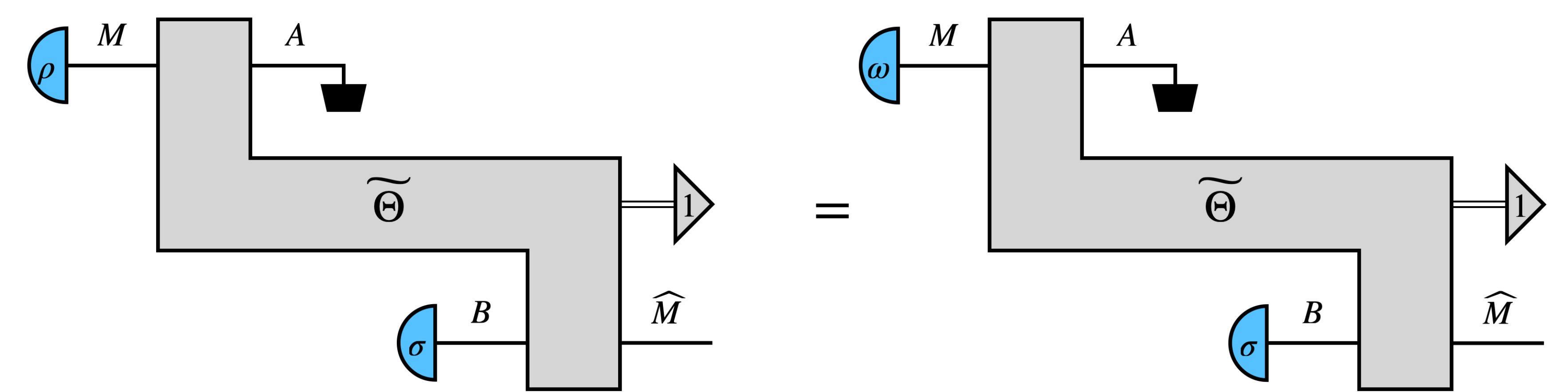}}
\vspace*{8pt}
\caption{The Alice-to-Bob-nonsignalling constraint on a bipartite subchannel $\bc{\Theta}_{MB\to A\g{M}}$ [see Eq.~\eqref{eq:NA}].  On coarse-graining Alice's output system $A$, no information can be transmitted from Alice's input system $M$ to Bob's output system $\g{M}$.}
\label{fig:pNA}
\end{figure}

\begin{proof}
See Appendix~\ref{app:pNA}.
\end{proof}

Proposition~\ref{prop:pNA} unifies three aspects of pNA protocols: (i) the compositional structure, (ii) the information-theoretic implication, and (iii) the resource-theoretic implication.

\begin{remark}[Equivalence of postselected protocols under rescaling]
\label{rem:rescaling}
According to Eqs.~\eqref{eq:error-classical} and \eqref{eq:error-quantum}, multiplying a probabilistic supermap $\Theta_{(A\to B)\to(M\to\g{M})}$ by a positive scalar has no influence on the conditional error probabilities.  Therefore, the equivalence between Proposition~\ref{prop:pNA}(i) and (iii) essentially proves the maximality of the class of pNA protocols in terms of being probabilistically replacement preserving.
\end{remark}

\begin{remark}[pNA protocols need not be Bob-to-Alice nonsignalling]
For conventional NA protocols, their equivalent bipartite channels are not only Alice-to-Bob nonsignalling but also Bob-to-Alice nonsignalling, and the latter property is implicitly required since these protocols are (deterministic) superchannels~\cite{chiribella_2008, gour_2019-2}.  However, pNA protocols are not necessarily Bob-to-Alice nonsignalling in general.  A simple (counter)example is the following pEA (and thus pNA) protocol represented by the triple $(\Phi_{A'B'},\ch{E}_{MA'\to A},\ch{D}_{BB'\to\g{M}})$ with $d_{A'}=d_{B'}=d_A=d_B$ and
\begin{alignat}{3}
	\ch{E}_{MA'\to A}\!\left[\rho_M\right]&:=\tr\!\left[\rho_M\right]\id_{A'\to A}\qquad\forall\rho_M, \\
	\ch{D}_{BB'\to\g{M}}\!\left[\rho_{BB'}\right]&:=\tr\!\left[\Phi_{BB'}\rho_{BB'}\right]\op{1}{1}_\g{M}\qquad\forall\rho_{BB'},
\end{alignat}
which can be verified to be Bob-to-Alice signalling.  Note that this protocol realises a postselected closed timelike curve, as put forward in Ref.~\citen{lloyd_2011}.
\end{remark}

\section{\texorpdfstring{One-shot \lowercase{p}EA \& \lowercase{p}NA capacities}{One-shot pEA \& pNA capacities}}
\label{sec:oneshot}

In this section, we establish the one-shot quantum and classical capacities of a channel in both pEA and pNA communication.  We derive upper and lower bounds on these capacities and prove that these bounds already match in the one-shot regime.  Our findings also suggest that nonsignalling correlations provide essentially no advantage over shared entanglement in terms of assisting postselected communication.

\subsection{Postselected teleportation-based coding}
\label{sec:teleportation}

We start by proposing a pEA communication scheme termed \emph{postselected teleportation-based coding}, which demonstrates how communication strategies can be transformed by simply allowing postselection on Bob's side.  The design of this scheme is, in spirit, distinct from that of all other EA schemes that do not leverage postselection~\cite{bennett_2002, shor_2004-1, hsieh_2008, anshu_2019-1, qi_2018}, and we will show that taking advantage of postselection can indeed enhance the performance in a communication task.

\begin{protocol}[Postselected teleportation-based coding]
\label{prot:teleportation}
Let $\ch{N}_{A\to B}$ be a quantum channel connecting Alice and Bob.  The coding scheme operates according to the following procedure (see Fig.~\ref{fig:teleportation}).
\begin{enumerate}
	\item \emph{Shared entanglement.} Alice and Bob share a maximally entangled state $\Phi_{A'B'}\equiv\Phi_{A_1'B_1'}\otimes\Phi_{A_2'B_2'}$, with Alice holding the system $A'\equiv A_1'A_2'$ and Bob holding the system $B'\equiv B_1'B_2'$.
	\item \emph{Encoding.} Alice measures the systems $MA_1'$ according to the POVM $\{\Phi_{MA_1'},\1_{MA_1'}-\Phi_{MA_1'}\}$, with the outcome $\Phi_{MA_1'}$ indicating a success.  If she succeeds, she applies a channel $\ch{P}_{A_2'\to A}$; otherwise, she applies a channel $\ch{Q}_{A_2'\to A}$.
	\item \emph{Transmission.} The system $A$ is passed from Alice to Bob through the channel $\ch{N}_{A\to B}$.
	\item \emph{Decoding.} Bob measures the systems $BB_2'$ according to a POVM $\{O_{BB_2'},\1_{BB_2'}-O_{BB_2'}\}$, where $0\leq O_{BB_2'}\leq\1_{BB_2'}$.  If he obtains the outcome $O_{BB_2'}$, he applies the identity channel $\id_{B_1'\to\g{M}}$ and announces to be conclusive; otherwise, he announces to be inconclusive.
\end{enumerate}
\end{protocol}

\begin{figure}[t]
\centerline{\includegraphics[scale=0.24]{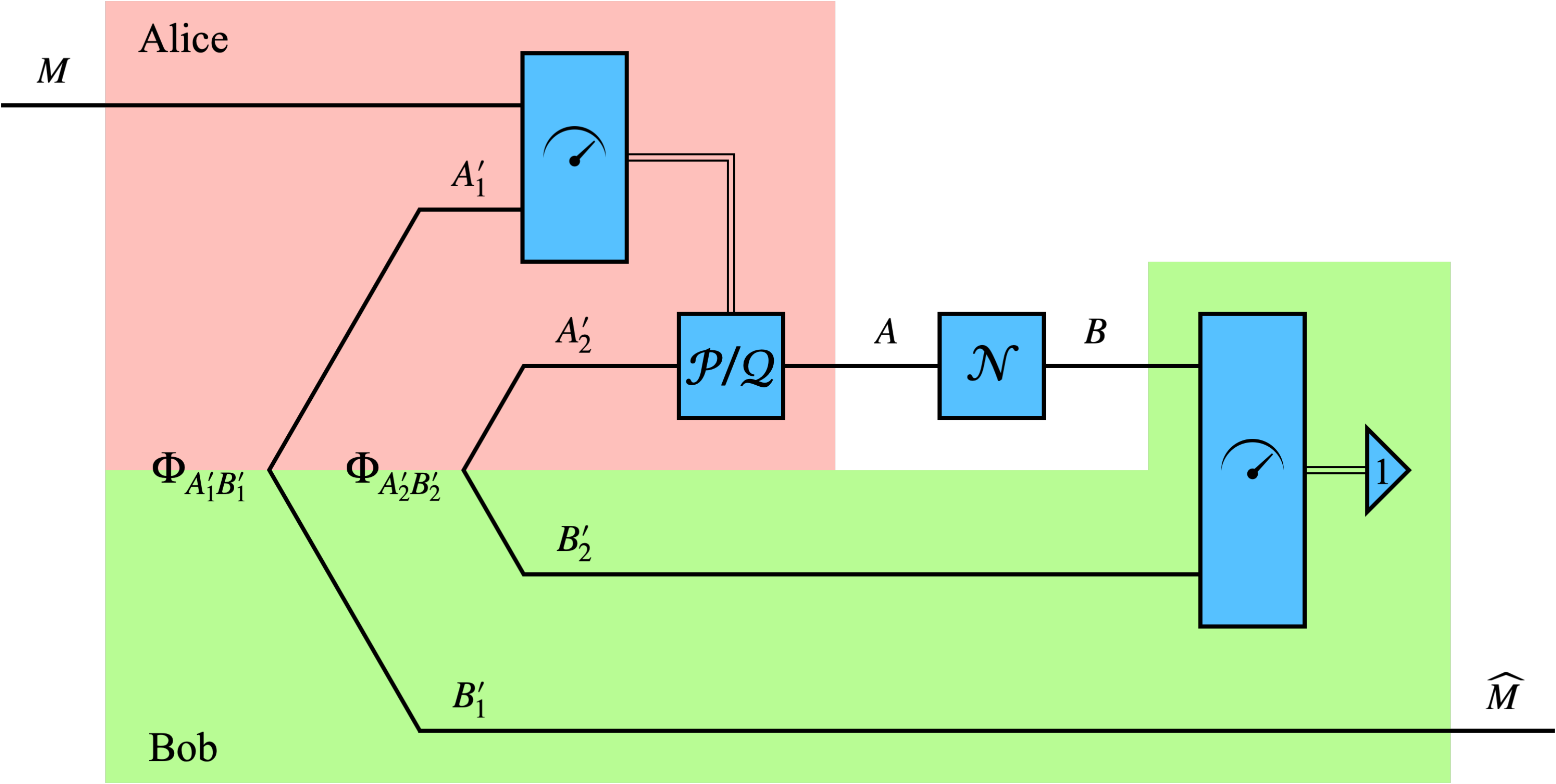}}
\vspace*{8pt}
\caption{Postselected teleportation-based coding over a channel $\ch{N}_{A\to B}$ (see Protocol~\ref{prot:teleportation}).  The red and green regions represent Alice's and Bob's operations, respectively.}
\label{fig:teleportation}
\end{figure}

Protocol~\ref{prot:teleportation} can be interpreted as follows.  In the encoding phase, Alice first tries to direct her message from $M$ to $B_1'$ by conducting probabilistic teleportation with the (incomplete) Bell measurement on $MA_1'$ and the maximally entangled state between $A_1'$ and $B_1$.  Then she aims at informing Bob whether her teleportation was successful or not by encoding this one bit of classical information into the quantum system $A$ and passing it to Bob through the channel $\ch{N}_{A\to B}$.  In the decoding phase, Bob tries to recover the one bit of classical information stored in $B$, with the assistance of $B_2'$, by conducting state discrimination.  He announces that he is conclusive if and only if he believes that Alice was trying to tell him that her teleportation was successful, in which case he would naturally think of $B_1'$ as and ideal recovery of the original message in $M$.

The following proposition provides a quantitative analysis of the performance of Protocol~\ref{prot:teleportation} in postselected communication.

\begin{proposition}[Achievability in one-shot pEA communication]
\label{prop:teleportation}
Let $\ch{N}_{A\to B}$ be a quantum channel, and let $\varepsilon\in(0,1)$.  Then there exists a pEA $(d_M,\varepsilon)$ quantum protocol over $\ch{N}_{A\to B}$ if
\begin{align}
\label{eq:teleportation}
	d_M^2&<\frac{\varepsilon}{1-\varepsilon}2^{\Delta(\ch{N})}+1,
\end{align}
where the measure $\Delta$ for a channel $\ch{N}_{A\to B}$ is defined as
\begin{align}
\label{eq:delta}
	\Delta(\ch{N})&:=\sup_{\ch{P}_{S\to A},\ch{Q}_{S\to A}}D_{\max}(\ch{N}\circ\ch{P}\|\ch{N}\circ\ch{Q}),
\end{align}
and the supremum is over all systems $S$ and channels $\ch{P}_{S\to A}$ and $\ch{Q}_{S\to A}$.
\end{proposition}

\begin{proof}
It suffices to show that there exists a postselected teleportation-based coding scheme that fulfils the requirement in Proposition~\ref{prop:teleportation}.  Let $M$ be a quantum system satisfying Eq.~\eqref{eq:teleportation}.  Let $\Theta_{(A\to B)\to(M\to\g{M})}^{\ch{P},\ch{Q},O}$ be a probabilistic supermap as described in Protocol~\ref{prot:teleportation}, parametrised by two channels $\ch{P}_{A_2'\to A}$ and $\ch{Q}_{A_2'\to A}$ and an operator $O_{BB_2'}$ satisfying $0\leq O_{BB_2'}\leq\1_{BB_2'}$.  By the principle of teleportation, we know that the simulated channel from $M$ to $\g{M}$ is the noiseless channel $\id_{M\to\g{M}}$ if Alice's measurement in the encoding phase succeeds.  Therefore, the worst-case error probability of quantum communication, conditioned on Bob being conclusive, is bounded from above by
\begin{align}
\label{pf:teleportation-1}
	&P_\abb{err}^\abb{(q)}(\Theta^{\ch{P},\ch{Q},O};\ch{N}) \notag\\
	&\qquad\leq\Pr\left\{\abb{Alice fails}\,\middle|\,\abb{Bob is conclusive}\right\} \notag\\
	&\qquad=\frac{\Pr\left\{\abb{Alice fails, Bob is conclusive}\right\}}{\Pr\left\{\abb{Alice succeeds, Bob is conclusive}\right\}+\Pr\left\{\abb{Alice fails, Bob is conclusive}\right\}}.
\end{align}
According to Protocol~\ref{prot:teleportation}, the following observations can be made:
\begin{align}
	\Pr\left\{\abb{Alice succeeds}\right\}&=1-\Pr\left\{\abb{Alice fails}\right\}=\frac{1}{d_M^2}, \label{pf:teleportation-2}
\end{align}
and
\begin{align}
	\Pr\left\{\abb{Bob is conclusive}\,\middle|\,\abb{Alice succeeds}\right\}&=\tr\!\left[O_{BB_2'}\left(\ch{N}_{A\to B}\circ\ch{P}_{A_2'\to A}\right)\!\left[\Phi_{A_2'B_2'}\right]\right] \notag\\
	&=\tr\!\left[O_{B_2'B}\Phi_{B_2'B}^{\ch{N}\circ\ch{P}}\right],
\end{align}
and
\begin{align}
	\Pr\left\{\abb{Bob is conclusive}\,\middle|\,\abb{Alice fails}\right\}&=\tr\!\left[O_{BB_2'}\left(\ch{N}_{A\to B}\circ\ch{Q}_{A_2'\to A}\right)\!\left[\Phi_{A_2'B_2'}\right]\right] \notag\\
	&=\tr\!\left[O_{B_2'B}\Phi_{B_2'B}^{\ch{N}\circ\ch{Q}}\right]. \label{pf:teleportation-3}
\end{align}
Inserting Eqs.~\eqref{eq:teleportation} and \eqref{pf:teleportation-2}--\eqref{pf:teleportation-3} to Eq.~\eqref{pf:teleportation-1}, and applying the Bayes' rule, we have that
\begin{align}
	P_\abb{err}^\abb{(q)}(\Theta^{\ch{P},\ch{Q},O};\ch{N})&\leq\frac{\left(1-\frac{1}{d_M^2}\right)\tr\!\left[O_{B_2'B}\Phi_{B_2'B}^{\ch{N}\circ\ch{Q}}\right]}{\frac{1}{d_M^2}\tr\!\left[O_{B_2'B}\Phi_{B_2'B}^{\ch{N}\circ\ch{P}}\right]+\left(1-\frac{1}{d_M^2}\right)\tr\!\left[O_{B_2'B}\Phi_{B_2'B}^{\ch{N}\circ\ch{Q}}\right]} \notag\\
	&=\left(\frac{1}{d_M^2-1}\frac{\tr\!\left[O_{B_2'B}\Phi_{B_2'B}^{\ch{N}\circ\ch{P}}\right]}{\tr\!\left[O_{B_2'B}\Phi_{B_2'B}^{\ch{N}\circ\ch{Q}}\right]}+1\right)^{-1} \notag\\
	&<\left(\left(\frac{1}{\varepsilon}-1\right)2^{-\Delta(\ch{N})}\frac{\tr\!\left[O_{B_2'B}\Phi_{B_2'B}^{\ch{N}\circ\ch{P}}\right]}{\tr\!\left[O_{B_2'B}\Phi_{B_2'B}^{\ch{N}\circ\ch{Q}}\right]}+1\right)^{-1}. \label{pf:teleportation-4}
\end{align}
Employing Eqs.~\eqref{eq:max-relative-choi} and \eqref{eq:max-relative}, the quantity $\Delta(\ch{N})$ defined in Eq.~\eqref{eq:delta} can be recast as
\begin{align}
\label{eq:delta-dual}
	\Delta(\ch{N})&=\log_2\sup_{\ch{P}_{S\to A},\ch{Q}_{S\to A},O_{TB}}\frac{\tr\!\left[O_{TB}\Phi_{TB}^{\ch{N}\circ\ch{P}}\right]}{\tr\!\left[O_{TB}\Phi_{TB}^{\ch{N}\circ\ch{Q}}\right]},
\end{align}
where the supremum is over all systems $S$, channels $\ch{P}_{S\to A}$ and $\ch{Q}_{S\to A}$, and operators $O_{TB}$ with $0\leq O_{TB}\leq\1_{TB}$ and $d_T=d_S$.  Fixing $A_2'=S$ and $B_2'=T$, optimising over $\ch{P}_{A_2'\to A}$, $\ch{Q}_{A_2'\to A}$, and $O_{B_2'B}$ on both sides of Eq.~\eqref{pf:teleportation-4}, and inserting Eq.~\eqref{eq:delta-dual}, we have that
\begin{align}
	&\inf_{\ch{P}_{A_2'\to A},\ch{Q}_{A_2'\to A},O_{B_2'B}}P_\abb{err}^\abb{(q)}(\Theta^{\ch{P},\ch{Q},O};\ch{N}) \notag\\
	&\qquad<\left(\left(\frac{1}{\varepsilon}-1\right)2^{-\Delta(\ch{N})}\sup_{\ch{P}_{A_2'\to A},\ch{Q}_{A_2'\to A},O_{B_2'B}}\frac{\tr\!\left[O_{B_2'B}\Phi_{B_2'B}^{\ch{N}\circ\ch{P}}\right]}{\tr\!\left[O_{B_2'B}\Phi_{B_2'B}^{\ch{N}\circ\ch{Q}}\right]}+1\right)^{-1} \notag\\*
	&\qquad=\varepsilon.
\end{align}
This implies that there exist two channels $\ch{P}_{A_2'\to A}$ and $\ch{Q}_{A_2'\to A}$ and an operator $O_{BB_2'}$ such that $0\leq O_{BB_2'}\leq\1_{BB_2'}$ and $P_\abb{err}^\abb{(q)}(\Theta^{\ch{P},\ch{Q},O};\ch{N})\leq\varepsilon$, making $\Theta_{(A\to B)\to(M\to\g{M})}^{\ch{P},\ch{Q},O}$ a $(d_M,\varepsilon)$ quantum protocol.
\end{proof}

\begin{remark}[On the case of an infinite one-shot pEA quantum capacity]
\label{rem:infinite}
Proposition~\ref{prop:teleportation} implies a lower bound on the one-shot pEA quantum capacity of $\ch{N}_{A\to B}$ in terms of the quantity $\Delta(\ch{N})$ defined in Eq.~\eqref{eq:delta}.  Due to the fact that $D_{\max}(\rho\|\sigma)=\infty$ if $\supp(\rho)\not\subseteq\supp(\sigma)$, every channel $\ch{N}_{A\to B}$ capable of producing output states with different supports  (e.g., any classical bit erasure channel with erasure probability $p<1$, or simply any unitary channel) has $\Delta(\ch{N})=\infty$ and thus an infinite one-shot $\varepsilon$-error pEA quantum capacity for all $\varepsilon\in(0,1)$. However, $\Delta(\ch{N})$ is finite for generic noisy channels.
\end{remark}

\subsection{One-shot pEA \& pNA quantum capacities}
\label{sec:oneshot-quantum}

Proposition~\ref{prop:teleportation} provides a lower bound on the one-shot pEA quantum capacity of a channel $\ch{N}_{A\to B}$ in terms of a novel quantity $\Delta(\ch{N})$, and it naturally brings up the following questions: What is the relation between $\Delta(\ch{N})$ and other known information-theoretic quantities? How far is this lower bound from the actual one-shot pEA quantum capacity $Q_\sa{pEA}^\varepsilon(\ch{N})$?  In what follows, we provide exact and concise answers to the above questions, which also enable us to derive closed-form expressions for the one-shot pEA and pNA quantum capacities.

\begin{lemma}[Characterisation of $\Delta$]
\label{lem:delta}
Let $\ch{N}_{A\to B}$ be a quantum channel.  Then
\begin{align}
	\Delta(\ch{N})&=I_\Omega(\ch{N}).
\end{align}
\end{lemma}

\begin{proof}
We first show $\Delta(\ch{N})\leq I_\Omega(\ch{N})$.  Following Eq.~\eqref{eq:delta-dual},
\begin{align}
	\Delta(\ch{N})&=\log_2\sup_{\ch{P}_{S\to A},\ch{Q}_{S\to A},O_{TB}}\frac{\tr\!\left[O_{TB}\Phi_{TB}^{\ch{N}\circ\ch{P}}\right]}{\tr\!\left[O_{TB}\Phi_{TB}^{\ch{N}\circ\ch{Q}}\right]} \notag\\
	&=\log_2\sup_{\ch{P}_{S\to A},\ch{Q}_{S\to A},O_{TB}}\frac{\tr\!\left[O_{TB}\left(\ch{N}_{A\to B}\circ\ch{P}_{S\to A}\right)\!\left[\Phi_{TS}\right]\right]}{\tr\!\left[O_{TB}\left(\ch{N}_{A\to B}\circ\ch{Q}_{S\to A}\right)\!\left[\Phi_{TS}\right]\right]} \notag\\
	&=\log_2\sup_{\ch{P}_{T\to R},\ch{Q}_{T\to R},O_{TB}}\frac{\tr\!\left[O_{TB}\left(\ch{P}_{R\to T}^\top\otimes\ch{N}_{A\to B}\right)\!\left[\Phi_{RA}\right]\right]}{\tr\!\left[O_{TB}\left(\ch{Q}_{R\to T}^\top\otimes\ch{N}_{A\to B}\right)\!\left[\Phi_{RA}\right]\right]} \notag\\
	&=\log_2\sup_{\ch{P}_{T\to R},\ch{Q}_{T\to R},O_{TB}}\frac{\tr\!\left[\ch{P}_{T\to R}^*\!\left[O_{TB}\right]\Phi_{RB}^\ch{N}\right]}{\tr\!\left[\ch{Q}_{T\to R}^*\!\left[O_{TB}\right]\Phi_{RB}^\ch{N}\right]}, \label{pf:delta-1}
\end{align}
where the supremum is over all channels $\ch{P}_{S\to A}$ and $\ch{Q}_{S\to A}$ and positive operators $O_{TB}$ with $d_T=d_S$, and $\ch{P}_{R\to T}^\top$ and $\ch{P}_{T\to R}^*$ denote the transpose and complex conjugate of $\ch{P}_{T\to R}$, respectively.  Note that we have dropped the redundant condition $O_{TB}\leq\1_{TB}$ for Eq.~\eqref{pf:delta-1}, since dividing $O_{TB}$ by its operator norm always fits such a condition while not changing the value of the objective function.  Here the third line of Eq.~\eqref{pf:delta-1} used the transpose trick for channels~\cite{gour_2019-2, khatri_2024}, and the last line used the property $(\ch{P}_{R\to T}^\top)^\dagger=\ch{P}_{T\to R}^*$.  Since $\ch{P}_{T\to R}^*$ and $\ch{Q}_{T\to R}^*$ are both channels (i.e., CPTP maps), we know that $(\ch{P}_{T\to R}^*\otimes\id_B)[O_{TB}]$ and $(\ch{Q}_{T\to R}^*\otimes\id_B)[O_{TB}]$ are positive semidefinite operators satisfying
\begin{align}
\label{pf:delta-2}
	\tr_R\!\left[\ch{P}_{T\to R}^*\!\left[O_{TB}\right]\right]&=\tr_R\!\left[\ch{Q}_{T\to R}^*\!\left[O_{TB}\right]\right].
\end{align}
By Lemma~\ref{lem:projective-mutual-information-dual} in Appendix~\ref{app:equivalent}, the dual formulation of $I_\Omega(\ch{N})$ is given by
\begin{align}
\label{eq:projective-mutual-channel-dual}
	I_\Omega(\ch{N})&=\log_2\sup_{P_{RB},Q_{RB}\geq0}\left\{\frac{\tr\!\left[P_{RB}\Phi_{RB}^\ch{N}\right]}{\tr\!\left[Q_{RB}\Phi_{RB}^\ch{N}\right]}\colon P_B=Q_B\right\}.
\end{align}
Following Eq.~\eqref{pf:delta-1} and employing Eqs.~\eqref{pf:delta-2} and \eqref{eq:projective-mutual-channel-dual}, we conclude that $\Delta(\ch{N})\leq I_\Omega(\ch{N})$.

For the other direction, consider an arbitrary pair of positive semidefinite operators $(P_{RB},Q_{RB})$ such that $P_B=Q_B$.  We note that the operator
\begin{align}
	O_{TB}&:=d_BP_B^\frac{1}{2}\Phi_{TB}P_B^\frac{1}{2}
\end{align}
is a purification of $P_B$ in the sense that it is rank one and satisfies $O_B=P_B$.  This implies that there exist two channels $\ch{P}_{T\to R}^P$ and $\ch{Q}_{T\to R}^Q$ such that $P_{RB}=\ch{P}_{T\to R}\!\left[O_{TB}\right]$ and $Q_{RB}=\ch{Q}_{T\to R}\!\left[O_{TB}\right]$.  Employing Eq.~\eqref{pf:delta-1}, we have that
\begin{align}
	\log_2\frac{\tr\!\left[P_{RB}\Phi_{RB}^\ch{N}\right]}{\tr\!\left[Q_{RB}\Phi_{RB}^\ch{N}\right]}&=\log_2\frac{\tr\!\left[\ch{P}_{T\to R}\!\left[O_{TB}\right]\Phi_{RB}^\ch{N}\right]}{\tr\!\left[\ch{Q}_{T\to R}\!\left[O_{TB}\right]\Phi_{RB}^\ch{N}\right]}\leq\Delta(\ch{N}).
\end{align}
Since this holds for every pair of positive semidefinite operators $(P_{RB},Q_{RB})$ satisfying $P_B=Q_B$, it follows from Eq.~\eqref{eq:projective-mutual-channel-dual} that $I_\Omega(\ch{N})\leq\Delta(\ch{N})$.
\end{proof}

\begin{remark}[Bounded dimensionality of the system $S$]
\label{rem:bounded}
The proof of Lemma~\ref{lem:delta} implies that the supremum in Eq.~\eqref{eq:delta} can be achieved by channels with an input system $S$ of a bounded dimensionality $d_S=d_B$.  Accordingly, the performance of Protocol~\ref{prot:teleportation} does not improve by enlarging the systems $A_2'$ and $B_2'$ beyond $d_{A_2'}=d_{B_2'}=d_B$.  Note that the size of these systems being bounded is not assumed when we initially defined $\Delta(\ch{N})$ or Protocol~\ref{prot:teleportation}.
\end{remark}

According to Lemma~\ref{lem:delta}, we readily obtain a lower bound on the one-shot pEA quantum capacity $Q_\sa{pEA}^\varepsilon(\ch{N})$ in terms of the projective mutual information $I_\Omega(\ch{N})$ and the conditional error probability $\varepsilon$.  In what follows, we establish an upper bound on the one-shot pNA quantum capacity $Q_\sa{pNA}^\varepsilon(\ch{N})$ that matches the lower bound on $Q_\sa{pEA}^\varepsilon(\ch{N})$.  As $\sa{pEA}\subset\sa{pNA}$ implies $Q_\sa{pEA}^\varepsilon(\ch{N})\leq Q_\sa{pNA}^\varepsilon(\ch{N})$, such an upper bound results in an approximate expression for both capacities.

\begin{theorem}[One-shot pEA \& pNA quantum capacities]
\label{thm:oneshot-quantum}
Let $\ch{N}_{A\to B}$ be a quantum channel, and let $\varepsilon\in(0,1)$.  Then the one-shot $\varepsilon$-error pEA and pNA quantum capacities satisfy
\begin{align}
	\log_2\!\left\lceil\sqrt{\frac{\varepsilon}{1-\varepsilon}2^{I_\Omega(\ch{N})}+1}-1\right\rceil&\leq Q_\sa{pEA}^\varepsilon(\ch{N})\leq Q_\sa{pNA}^\varepsilon(\ch{N})\leq\log_2\!\left\lfloor\sqrt{\frac{\varepsilon}{1-\varepsilon}2^{I_\Omega(\ch{N})}+1}\right\rfloor.
\end{align}
\end{theorem}

\begin{proof}
Due to the fact $Q_\sa{pEA}^\varepsilon(\ch{N})\leq Q_\sa{pNA}^\varepsilon(\ch{N})$, it suffices to prove the lower bound on $Q_\sa{pEA}^\varepsilon(\ch{N})$ and the upper bound on $Q_\sa{pNA}^\varepsilon(\ch{N})$.  By Proposition~\ref{prop:teleportation} and Lemma~\ref{lem:delta}, there exists a pEA $(d_M,\varepsilon)$ quantum protocol over $\ch{N}_{A\to B}$ if
\begin{align}
	d_M&=\left\lceil\sqrt{\frac{\varepsilon}{1-\varepsilon}2^{I_\Omega(\ch{N})}+1}-1\right\rceil.
\end{align}
The lower bound on $Q_\sa{pEA}^\varepsilon(\ch{N})$ thus follows.

For the upper bound, consider a general pNA $(d_M,\varepsilon)$ quantum protocol $\Theta_{(A\to B)\to(M\to\g{M})}$ over $\ch{N}_{A\to B}$.  Let $\ch{N}_{M\to\g{M}}'\equiv\Theta_{(A\to B)\to(M\to\g{M})}\{\ch{N}_{A\to B}\}$ be the subchannel simulated with the protocol over $\ch{N}_{A\to B}$.  By Lemma~\ref{lem:projective-mutual-information-choi} in Appendix~\ref{app:equivalent}, and inserting Eqs.~\eqref{eq:projective-relative} and \eqref{eq:max-relative-choi} and the first line of Eq.~\eqref{eq:max-relative}, the quantity $I_\Omega(\ch{N})$ can be recast as
\begin{align}
	I_\Omega(\ch{N})&=\inf_{\ch{R}_{A\to B}^\sigma\in\s{R}}D_\Omega(\Phi^\ch{N}\|\Phi^{\ch{R}^\sigma}) \notag\\
	&=\log_2\inf_{\lambda,\mu,\sigma_B}\left\{\lambda\mu\colon\Phi^\ch{N}\leq\lambda\Phi^{\ch{R}^\sigma},\;\Phi^{\ch{R}^\sigma}\leq\mu\Phi^\ch{N}\right\}, \label{pf:oneshot-quantum-1}
\end{align}
where the infimum on the second line is over all real numbers $\lambda$ and $\mu$ and all states $\sigma_B$.  Consider an arbitrary feasible solution $(\lambda,\mu,\sigma_B)$ to the optimisation in Eq.~\eqref{pf:oneshot-quantum-1}, which satisfies
\begin{align}
	\Phi^\ch{N}&\leq\lambda\Phi^{\ch{R}^\sigma}, \label{pf:oneshot-quantum-2}\\
	\Phi^{\ch{R}^\sigma}&\leq\mu\Phi^\ch{N}. \label{pf:oneshot-quantum-3}
\end{align}
By Proposition~\ref{prop:pNA}, the pNA protocol is probabilistically replacement preserving, and thus there exists a state $\sigma_\g{M}'$ and $p\in[0,1]$ such that $\Theta_{(A\to B)\to(M\to\g{M})}\{\ch{R}_{A\to B}^\sigma\}=p\ch{R}_{M\to\g{M}}^{\sigma'}$.  Since the protocol, as a probabilistic supermap, preserves complete positivity, it follows from Eqs.~\eqref{pf:oneshot-quantum-2} and \eqref{pf:oneshot-quantum-3} that
\begin{align}
	\Phi_{R\g{M}}^{\ch{N}'}&\leq\lambda\Phi_{R\g{M}}^{\ch{R}^{\sigma'}}=\frac{\lambda p}{d_M}\1_R\otimes\sigma_\g{M}', \label{pf:oneshot-quantum-5}\\
	\Phi_{R\g{M}}^{\ch{N}'}&\geq\frac{1}{\mu}\Phi_{R\g{M}}^{\ch{R}^{\sigma'}}=\frac{p}{\mu d_M}\1_R\otimes\sigma_\g{M}'. \label{pf:oneshot-quantum-6}
\end{align}
Note that Eq.~\eqref{pf:oneshot-quantum-5} and the assumption of $\Theta_{(A\to B)\to(M\to\g{M})}$ being a $(d_M,\varepsilon)$ protocol over $\ch{N}_{A\to B}$ rule out the possibility of $p=0$.  Then by Eqs.~\eqref{eq:error-quantum} and \eqref{pf:oneshot-quantum-5}, we have that
\begin{align}
	1-\varepsilon&\leq1-P_\abb{err}^\abb{(q)}(\Theta;\ch{N}) \notag\\
	&\leq\frac{\tr\!\left[\Phi_{R\g{M}}\ch{N}_{M\to\g{M}}'\!\left[\Phi_{RM}\right]\right]}{\tr\!\left[\ch{N}_{M\to\g{M}}'\!\left[\Phi_{RM}\right]\right]} \notag\\
	&=\frac{\tr\!\left[\Phi_{R\g{M}}\Phi_{R\g{M}}^{\ch{N}'}\right]}{\tr\!\left[\Phi_{R\g{M}}^{\ch{N}'}\right]} \notag\\
	&\leq\frac{\lambda p\tr\!\left[\Phi_{R\g{M}}\left(\1_R\otimes\sigma_\g{M}'\right)\right]}{d_M\tr\!\left[\Phi_{R\g{M}}^{\ch{N}'}\right]} \notag\\
	&=\frac{\lambda p}{d_M^2\tr\!\left[\Phi_{R\g{M}}^{\ch{N}'}\right]}. \label{pf:oneshot-quantum-7}
\end{align}
On the other hand, by Eqs.~\eqref{eq:error-quantum} and \eqref{pf:oneshot-quantum-6}, we have that
\begin{align}
	\varepsilon&\geq P_\abb{err}^\abb{(q)}(\Theta;\ch{N}) \notag\\
	&\geq1-\frac{\tr\!\left[\Phi_{R\g{M}}\ch{N}_{M\to\g{M}}'\!\left[\Phi_{RM}\right]\right]}{\tr\!\left[\ch{N}_{M\to\g{M}}'\!\left[\Phi_{RM}\right]\right]} \notag\\
	&=\frac{\tr\!\left[\left(\1_{R\g{M}}-\Phi_{R\g{M}}\right)\Phi_{R\g{M}}^{\ch{N}'}\right]}{\tr\!\left[\Phi_{R\g{M}}^{\ch{N}'}\right]} \notag\\
	&\geq\frac{p\tr\!\left[\left(\1_{R\g{M}}-\Phi_{R\g{M}}\right)\left(\1_R\otimes\sigma_\g{M}'\right)\right]}{\mu d_M\tr\!\left[\Phi_{R\g{M}}^{\ch{N}'}\right]} \notag\\
	&=\left(1-\frac{1}{d_M^2}\right)\frac{p}{\mu\tr\!\left[\Phi_{R\g{M}}^{\ch{N}'}\right]}. \label{pf:oneshot-quantum-8}
\end{align}
Combining Eqs.~\eqref{pf:oneshot-quantum-7} and \eqref{pf:oneshot-quantum-8} yields that
\begin{align}
	d_M^2&\leq\frac{\varepsilon\lambda\mu}{1-\varepsilon}+1.
\end{align}
Since this holds for every feasible solution $(\lambda,\mu,\sigma_B)$ to the optimisation in Eq.~\eqref{pf:oneshot-quantum-1}, and since $d_M$ is an integer, it follows from Eq.~\eqref{pf:oneshot-quantum-1} that
\begin{align}
	d_M&\leq\left\lfloor\sqrt{\frac{\varepsilon}{1-\varepsilon}\inf_{\lambda,\mu,\sigma_B}\left\{\lambda\mu\colon\Phi^\ch{N}\leq\lambda\Phi^{\ch{R}^\sigma},\;\Phi^{\ch{R}^\sigma}\leq\mu\Phi^\ch{N}\right\}+1}\right\rfloor \notag\\
	&=\left\lfloor\sqrt{\frac{\varepsilon}{1-\varepsilon}2^{I_\Omega(\ch{N})}+1}\right\rfloor,
\end{align}
which establishes the upper bound.
\end{proof}

\begin{remark}[Tightness of Theorem~\ref{thm:oneshot-quantum}]
The upper and lower bounds in Theorem~\ref{thm:oneshot-quantum} coincide unless $\sqrt{\frac{\varepsilon}{1-\varepsilon}2^{I_\Omega(\ch{N})}+1}$ is an integer.  This suggests that the bounds we derived are the tightest possible analytic bounds in principle and that $Q_\sa{pEA}^\varepsilon(\ch{N})=Q_\sa{pNA}^\varepsilon(\ch{N})$ for almost all channels $\ch{N}_{A\to B}$ and all $\varepsilon\in(0,1)$.
\end{remark}

\begin{remark}[Alternative proof of the lower bound on $Q_\sa{pNA}^\varepsilon$]
In Appendix~\ref{app:achievability}, we provide an alternative proof of the lower bound $Q_\sa{pNA}^\varepsilon(\ch{N})\geq\log_2\!\left\lceil\sqrt{\frac{\varepsilon}{1-\varepsilon}2^{I_\Omega(\ch{N})}+1}-1\right\rceil$, where we construct a pNA protocol that achieves this lower bound without utilising postselected teleportation-based coding.
\end{remark}

\subsection{One-shot pEA \& pNA classical capacities}
\label{sec:oneshot-classical}

When entanglement assistance is available, every quantum protocol that transmits one quantum bit can be converted to a classical protocol that transmits two classical bits.  This is still the case in postselected communication, by applying the standard super-dense coding scheme~\cite{bennett_1992} to the noisy subchannel simulated with the quantum protocol.

\begin{lemma}[Classical communication via quantum protocols]
\label{lem:classical}
Let $\ch{N}_{A\to B}$ be a quantum channel, and let $\varepsilon\in(0,1)$.  If there exists a pNA $(d_M,\varepsilon)$ quantum protocol over $\ch{N}_{A\to B}$, then there exists a pNA $(d_M^2,\varepsilon)$ classical protocol over $\ch{N}_{A\to B}$.  If the quantum protocol is furthermore pEA, then so is the classical protocol.
\end{lemma}

\begin{proof}
Let $\ch{N}_{M\to\g{M}}'$ be the subchannel simulated with the $(d_M,\varepsilon)$ quantum protocol over $\ch{N}_{A\to B}$, where $M$ is a quantum system.  By Eq.~\eqref{eq:error-quantum}, for every pure state $\psi_{RM}$, we have that
\begin{align}
	\frac{\tr\!\left[\psi_{R\g{M}}\ch{N}_{M\to\g{M}}'\!\left[\psi_{RM}\right]\right]}{\tr\!\left[\ch{N}_{M\to\g{M}}'\!\left[\psi_{RM}\right]\right]}&\geq1-\varepsilon. \label{pf:classical-1}
\end{align}
Let $\cl{M}$ be a classical system with $d_\cl{M}=d_M^2$.  To prove Lemma~\ref{lem:classical}, it suffices to construct a pEA $(d_\cl{M},\varepsilon)$ classical protocol over $\ch{N}_{M\to\g{M}}'$, and we do so using super-dense coding~\cite{bennett_1992}.  Specifically, let Alice and Bob share a maximally entangled state $\Phi_{A'B'}$ with $d_{A'}=d_{B'}=d_M$.  Let Alice's encoding operation be the following channel:
\begin{align}
	\ch{E}_{\cl{M}A'\to M}\!\left[\rho_{\cl{M}A'}\right]&:=\sum_{\cl{m}\in[d_M^2]}W_{A'\to M}^\cl{m}\tr_\cl{M}\!\left[\op{\cl{m}}{\cl{m}}_\cl{M}\rho_{\cl{M}A'}\right](W_{A'\to M}^\cl{m})^\dagger\qquad\forall\rho_{\cl{M}A'}, \label{eq:classical-2}
\end{align}
where $W_{A'\to M}^\cl{m}$ is the $\cl{m}$th Heisenberg--Weyl operator for $\cl{m}\in[d_M^2]$ (see Ref.~\cite[Sec.~3.7.2]{wilde_2017} for a definition).  Let Bob's decoding operation be the following channel:
\begin{align}
	\ch{D}_{\g{M}B'\to\g{\cl{M}}}\!\left[\rho_{\g{M}B'}\right]&:=\sum_{\cl{m}\in[d_M^2]}\tr\!\left[\Phi_{\g{M}B'}^\cl{m}\rho_{\g{M}B'}\right]\op{\cl{m}}{\cl{m}}_\g{\cl{M}}\qquad\forall\rho_{\g{M}B'}, \label{eq:classical-3}
\end{align}
where $\Phi_{\g{M}B'}^\cl{m}:=W_\g{M}^\cl{m}\Phi_{\g{M}B'}(W_\g{M}^\cl{m})^\dagger$ denotes the $\cl{m}$th Bell state between $\g{M}$ and $B'$.  Let $\ch{N}_{\cl{M}\to\g{\cl{M}}}''$ be the classical subchannel simulated with the pEA classical protocol represented by the triple $(\Phi_{A'B'},\ch{E}_{\cl{M}A'\to M},\ch{D}_{\g{M}B'\to\g{\cl{M}}})$ over $\ch{N}_{M\to\g{M}}'$, namely,
\begin{align}
	\ch{N}_{\cl{M}\to\g{\cl{M}}}''\!\left[\rho_\cl{M}\right]&:=\left(\ch{D}_{\g{M}B'\to\g{\cl{M}}}\circ\ch{N}_{M\to\g{M}}'\circ\ch{E}_{\cl{M}A'\to M}\right)\!\left[\rho_\cl{M}\otimes\Phi_{A'B'}\right]\qquad\forall\rho_\cl{M}. \label{eq:classical-4}
\end{align}
Inserting Eqs.~\eqref{eq:classical-2} and \eqref{eq:classical-3} to Eq.~\eqref{eq:classical-4}, for all $\cl{m}\in[d_M^2]$, we have that
\begin{align}
	\frac{\tr\!\left[\op{\cl{m}}{\cl{m}}_\g{\cl{M}}\ch{N}_{\cl{M}\to\g{\cl{M}}}''\!\left[\op{\cl{m}}{\cl{m}}_\cl{M}\right]\right]}{\tr\!\left[\ch{N}_{\cl{M}\to\g{\cl{M}}}''\!\left[\op{\cl{m}}{\cl{m}}_\cl{M}\right]\right]}&=\frac{\tr\!\left[\Phi_{\g{M}B'}^\cl{m}\ch{N}_{M\to\g{M}}'\!\left[\Phi_{MB'}^\cl{m}\right]\right]}{\tr\!\left[\ch{N}_{M\to\g{M}}'\!\left[\Phi_{MB'}^\cl{m}\right]\right]}\geq1-\varepsilon. \label{pf:classical-5}
\end{align}
Here the last step of Eq.~\eqref{pf:classical-5} employed Eq.~\eqref{pf:classical-1} while substituting $R=B'$ and $\psi_{RM}=\Phi_{B'M}^\cl{m}$.  By Eq.~\eqref{eq:error-classical}, we complete the construction.
\end{proof}

With Lemma~\ref{lem:classical} in place, we know that the one-shot pEA and pNA classical capacities of a channel are at least twice as large as the corresponding quantum capacities.  Now we show that this factor of two is optimal by proving an upper bound on the one-shot pNA classical capacity.

\begin{theorem}[One-shot pEA \& pNA classical capacities]
\label{thm:oneshot-classical}
Let $\ch{N}_{A\to B}$ be a quantum channel, and let $\varepsilon\in(0,1)$.  Then the one-shot $\varepsilon$-error pEA and pNA classical capacities satisfy
\begin{align}
	\log_2\!\left\lceil\sqrt{\frac{\varepsilon}{1-\varepsilon}2^{I_\Omega(\ch{N})}+1}-1\right\rceil^2&\leq C_\sa{pEA}^\varepsilon(\ch{N})\leq C_\sa{pNA}^\varepsilon(\ch{N})\leq\log_2\!\left\lfloor\frac{\varepsilon}{1-\varepsilon}2^{I_\Omega(\ch{N})}+1\right\rfloor.
\end{align}
\end{theorem}

\begin{proof}
Due to the fact $C_\sa{pEA}^\varepsilon(\ch{N})\leq C_\sa{pNA}^\varepsilon(\ch{N})$ (since $\sa{pEA}\subset\sa{pNA}$), it suffices to prove the lower bound on $C_\sa{pEA}^\varepsilon(\ch{N})$ and the upper bound on $C_\sa{pNA}^\varepsilon(\ch{N})$.  By Lemma~\ref{lem:classical}, we can infer $C_\sa{pEA}^\varepsilon(\ch{N})\geq2Q_\sa{pEA}^\varepsilon(\ch{N})$, which gives rise to the lower bound on $C_\sa{pEA}^\varepsilon(\ch{N})$ according to Theorem~\ref{thm:oneshot-quantum}.

For the upper bound, consider a general pNA $(d_M,\varepsilon)$ classical protocol $\Theta_{(A\to B)\to(M\to\g{M})}$ over $\ch{N}_{A\to B}$, where $M$ is a classical system.  As a probabilistic supermap, $\Theta_{(A\to B)\to(M\to\g{M})}$ can be decomposed into a preprocessing channel $\ch{E}_{M\to AE}$ and a postprocessing subchannel $\ch{D}_{BE\to\g{M}}$, where $E$ is a side memory system~\cite[Theorem~2]{chiribella_2008} (also see Ref.~\cite[Sec.~5]{burniston_2020}).  Let $\ch{N}_{M\to\g{M}}'$ be the subchannel simulated with the protocol over $\ch{N}_{A\to B}$, so that
\begin{align}
	\ch{N}_{M\to\g{M}}'\!\left[\rho_M\right]&\equiv\left(\ch{D}_{BE\to\g{M}}\circ\ch{N}_{A\to B}\circ\ch{E}_{M\to AE}\right)\!\left[\rho_M\right]\qquad\forall\rho_M. \label{pf:oneshot-classical-0}
\end{align}
By Eq.~\eqref{eq:error-classical}, for all $m\in[d_M]$, we have that
\begin{align}
	\frac{\tr\!\left[\op{m}{m}_\g{M}\ch{N}_{M\to\g{M}}'\!\left[\op{m}{m}_M\right]\right]}{\tr\!\left[\ch{N}_{M\to\g{M}}'\!\left[\op{m}{m}_M\right]\right]}&\geq1-\varepsilon,
\end{align}
which through slight manipulation implies that
\begin{align}
	\frac{\sum_{m\in[d_M]}\tr\!\left[\op{m}{m}_\g{M}\ch{N}_{M\to\g{M}}'\!\left[\op{m}{m}_M\right]\right]}{\sum_{m\in[d_M]}\tr\!\left[\ch{N}_{M\to\g{M}}'\!\left[\op{m}{m}_M\right]\right]}&\geq1-\varepsilon. \label{pf:oneshot-classical-1}
\end{align}
Let $R$ be a classical reference system with $d_R=d_M$.  Let $\overline{\Phi}_{RM}\equiv\frac{1}{d_M}\sum_{m\in[d_M]}\op{m}{m}_R\otimes\op{m}{m}_M$ denote the standard maximally classically correlated state between $R$ and $M$, and let $\overline{\Pi}_{R\g{M}}\equiv\sum_{m\in[d_M]}\op{m}{m}_R\otimes\op{m}{m}_\g{M}$ denote the classical comparator projection.  Denote $\rho_{RAE}\equiv\ch{E}_{M\to AE}[\overline{\Phi}_{RM}]$ and
\begin{align}
	\omega_{R\g{M}}&\equiv\left(\ch{D}_{BE\to\g{M}}\circ\ch{N}_{A\to B}\right)\!\left[\rho_{RAE}\right] \notag\\
	&=\ch{N}_{M\to\g{M}}'\!\left[\overline{\Phi}_{RM}\right]. \label{pf:oneshot-classical-2}
\end{align}
Here the last line of Eq.~\eqref{pf:oneshot-classical-2} used Eq.~\eqref{pf:oneshot-classical-0}.  It follows from Eq.~\eqref{pf:oneshot-classical-1} that
\begin{align}
	\frac{\tr\!\left[\left(\1_{RM}-\overline{\Pi}_{R\g{M}}\right)\omega_{R\g{M}}\right]}{\tr\!\left[\omega_{R\g{M}}\right]}&=1-\frac{\tr\!\left[\overline{\Pi}_{R\g{M}}\ch{N}_{M\to\g{M}}'\!\left[\overline{\Phi}_{RM}\right]\right]}{\tr\!\left[\ch{N}_{M\to\g{M}}'\!\left[\overline{\Phi}_{RM}\right]\right]} \notag\\
	&=1-\frac{\sum_{m\in[d_M]}\tr\!\left[\op{m}{m}_\g{M}\ch{N}_{M\to\g{M}}'\!\left[\op{m}{m}_M\right]\right]}{\sum_{m\in[d_M]}\tr\!\left[\ch{N}_{M\to\g{M}}'\!\left[\op{m}{m}_M\right]\right]} \notag\\
	&\leq\varepsilon. \label{pf:oneshot-classical-3}
\end{align}
Let $\ch{R}_{A\to B}^\sigma$ be an arbitrary replacement channel, where $\sigma_B$ is a state.  By Proposition~\ref{prop:pNA}, the pNA protocol is probabilistically replacement preserving, and thus there exists a state $\sigma_\g{M}'$ and $p\in[0,1]$ such that $\Theta_{(A\to B)\to(M\to\g{M})}\{\ch{R}_{A\to B}^\sigma\}=p\ch{R}_{M\to\g{M}}^{\sigma'}$.  Denote
\begin{align}
	\tau_{R\g{M}}^\sigma&\equiv\left(\ch{D}_{BE\to\g{M}}\circ\ch{R}_{A\to B}^\sigma\right)\!\left[\rho_{RAE}\right] \label{pf:oneshot-classical-4} \notag\\
	&=p\ch{R}_{M\to\g{M}}^{\sigma'}\!\left[\overline{\Phi}_{RM}\right] \notag\\
	&=\frac{p}{d_M}\1_R\otimes\sigma_{\g{M}}',
\end{align}
where we used $d_R=d_M$.  Then it follows that
\begin{align}
	\frac{\tr\!\left[\overline{\Pi}_{R\g{M}}\tau_{R\g{M}}^\sigma\right]}{\tr\!\left[\tau_{R\g{M}}^\sigma\right]}&=\frac{\tr\!\left[\overline{\Pi}_{R\g{M}}\left(\1_R\otimes\sigma_{\g{M}}'\right)\right]}{\tr\!\left[\1_R\otimes\sigma_{\g{M}}'\right]}=\frac{1}{d_M}. \label{pf:oneshot-classical-5}
\end{align}
Since the postselected hypothesis testing relative entropy in Eq.~\eqref{eq:pH-relative} can be defined on subnormalised states as well by rescaling, we have that
\begin{align}
	D_\abb{pH}^\varepsilon(\omega\|\tau^\sigma)&:=D_\abb{pH}^\varepsilon(\omega/\tr\!\left[\omega\right]\|\tau^\sigma/\tr\!\left[\tau^\sigma\right]) \notag\\
	&\hphantom{:}=-\log_2\inf_{P,Q\geq0}\left\{\frac{\tr\!\left[P\tau^\sigma\right]}{\tr\!\left[\left(P+Q\right)\tau^\sigma\right]}\colon\frac{\tr\!\left[Q\omega\right]}{\tr\!\left[\left(P+Q\right)\omega\right]}\leq\varepsilon,\;P+Q\leq\1\right\}. \label{pf:oneshot-classical-6}
\end{align}
By Eq.~\eqref{pf:oneshot-classical-3},
\begin{align}
	\begin{cases}
		P_{R\g{M}}=\overline{\Pi}_{R\g{M}}, \\
		Q_{R\g{M}}=\1_{R\g{M}}-\overline{\Pi}_{R\g{M}}
	\end{cases}
\end{align}
is a feasible solution to the optimisation in Eq.~\eqref{pf:oneshot-classical-6}.  Therefore, by Eq.~\eqref{pf:oneshot-classical-5},
\begin{align}
	D_\abb{pH}^\varepsilon(\omega\|\tau^\sigma)&\geq-\log_2\frac{\tr\!\left[\overline{\Pi}_{R\g{M}}\tau_{R\g{M}}^\sigma\right]}{\tr\!\left[\tau_{R\g{M}}^\sigma\right]}=\log_2d_M.
\end{align}
Employing the data-processing inequality in Eq.~\eqref{eq:data-processing}, and by Eqs.~\eqref{pf:oneshot-classical-2} and \eqref{pf:oneshot-classical-4}, it follows that
\begin{align}
	\log_2d_M&\leq D_\abb{pH}^\varepsilon\big(\left(\ch{D}_{BE\to\g{M}}\circ\ch{N}_{A\to B}\right)\!\left[\rho_{RAE}\right]\big\|\left(\ch{D}_{BE\to\g{M}}\circ\ch{R}_{A\to B}^\sigma\right)\!\left[\rho_{RAE}\right]\big) \notag\\
	&\leq D_\abb{pH}^\varepsilon(\ch{N}_{A\to B}\!\left[\rho_{RAE}\right]\|\ch{R}_{A\to B}^\sigma\!\left[\rho_{RAE}\right]) \notag\\
	&=D_\abb{pH}^\varepsilon(\ch{N}_{A\to B}\!\left[\rho_{REA}\right]\|\rho_{RE}\otimes\sigma_B). \label{pf:oneshot-classical-7}
\end{align}
Since Eq.~\eqref{pf:oneshot-classical-7} holds for all states $\sigma_B$, we have that
\begin{align}
	\log_2d_M&\leq\inf_{\sigma_B}D_\abb{pH}^\varepsilon(\ch{N}_{A\to B}\!\left[\rho_{REA}\right]\|\rho_{RE}\otimes\sigma_B) \notag\\
	&\leq\sup_{\rho_{R'A}}\inf_{\sigma_B}D_\abb{pH}^\varepsilon(\ch{N}_{A\to B}\!\left[\rho_{R'A}\right]\|\rho_{R'}\otimes\sigma_B) \notag\\
	&=\log_2\!\left(\frac{\varepsilon}{1-\varepsilon}2^{\sup_{\rho_{R'A}}\inf_{\sigma_B}D_\Omega(\ch{N}_{A\to B}\!\left[\rho_{R'A}\right]\|\rho_{R'}\otimes\sigma_B)}+1\right) \notag\\
	&=\log_2\!\left(\frac{\varepsilon}{1-\varepsilon}2^{I_\Omega(\ch{N})}+1\right). \label{pf:oneshot-classical-9}
\end{align}
Here the third line of Eq.~\eqref{pf:oneshot-classical-9} employed Lemma~\ref{lem:pH-projective}, and the last line follows from Eqs.~\eqref{eq:projective-mutual} and \eqref{eq:projective-mutual-channel}.  Since $d_M$ is an integer, the upper bound follows.
\end{proof}

\begin{remark}[Tightness of Theorem~\ref{thm:oneshot-classical}]
The upper and lower bounds in Theorem~\ref{thm:oneshot-classical} coincide if $k^2<\frac{\varepsilon}{1-\varepsilon}2^{I_\Omega(\ch{N})}+1<k^2+1$ for some integer $k$.
\end{remark}

\begin{remark}[Alternative derivation of the upper bound on $Q_\sa{pNA}^\varepsilon$]
The upper bound in Theorem~\ref{thm:oneshot-quantum} can be recovered from the upper bound in Theorem~\ref{thm:oneshot-classical}.  Specifically, by Lemma~\ref{lem:classical}, we have $C_\sa{pNA}^\varepsilon(\ch{N})\geq2Q_\sa{pNA}^\varepsilon(\ch{N})$, which implies $Q_\sa{pNA}^\varepsilon(\ch{N})\leq\log_2\!\sqrt{\big\lfloor\frac{\varepsilon}{1-\varepsilon}2^{I_\Omega(\ch{N})}+1\big\rfloor}$ by Theorem~\ref{thm:oneshot-classical}.  Since $2^{Q_\sa{pEA}^\varepsilon(\ch{N})}$ is an integer by definition, we necessarily have $Q_\sa{pNA}^\varepsilon(\ch{N})\leq\log_2\!\left\lfloor\sqrt{\frac{\varepsilon}{1-\varepsilon}2^{I_\Omega(\ch{N})}+1}\right\rfloor$, due to the property $\lfloor\sqrt{r}\rfloor\leq\sqrt{\lfloor r\rfloor}<\lfloor\sqrt{r}\rfloor+1$ for all $r\geq0$.
\end{remark}

\begin{remark}[Alternative approach to an upper bound on $C_\sa{pNA}^\varepsilon$]
The upper bound in Theorem~\ref{thm:oneshot-quantum} also implies the following upper bound on $C_\sa{pNA}^\varepsilon(\ch{N})$ by applying the standard teleportation scheme~\cite{bennett_1992} to a simulated classical subchannel:
\begin{align}
    C_\sa{pNA}^\varepsilon(\ch{N})&\leq\log_2\!\left(\left\lfloor\sqrt{\frac{\varepsilon}{1-\varepsilon}2^{I_\Omega(\ch{N})}+1}+1\right\rfloor^2-1\right).
\end{align}
However, this upper bound is slightly looser than the one provided in Theorem~\ref{thm:oneshot-classical}, despite their asymptotic equivalence.
\end{remark}

To conclude this section, we informally summarise Theorems~\ref{thm:oneshot-quantum} and \ref{thm:oneshot-classical} in one line by disregarding the floor and ceiling operations as follows:
\begin{align}
	C_\sa{pEA}^\varepsilon(\ch{N})&\approx C_\sa{pNA}^\varepsilon(\ch{N})\approx2Q_\sa{pEA}^\varepsilon(\ch{N})\approx2Q_\sa{pNA}^\varepsilon(\ch{N})\approx\log_2\!\left(\frac{\varepsilon}{1-\varepsilon}2^{I_\Omega(\ch{N})}+1\right).
\end{align}

\section{\texorpdfstring{Asymptotic \lowercase{p}EA \& \lowercase{p}NA capacities}{Asymptotic pEA \& pNA capacities}}
\label{sec:asymptotic}

In this section, we derive the pEA and pNA capacities and their strong converses of a channel in the asymptotic regime.  As it turns out, both the capacities and their strong converses have a single-letter characterisation in terms of the channel's projective mutual information.  This result is in form reminiscent of Shannon's noisy-channel coding theorem in the classical setting~\cite{shannon_1948} and the entanglement-assisted capacity theorem in the conventional quantum setting without postselection~\cite{bennett_2002, holevo_2002}.  The simplicity of the result is largely due to the fact that the projective mutual information is additive, as stated in the following lemma.

\begin{lemma}[Additivity of $I_\Omega$]
\label{lem:additivity}
Let $\ch{N}_{A\to B}$ and $\ch{N}_{A'\to B'}'$ be two quantum channels.  Then
\begin{align}
	I_\Omega(\ch{N}\otimes\ch{N}')&=I_\Omega(\ch{N})+I_\Omega(\ch{N}').
\end{align}
\end{lemma}

\begin{proof}
Since the set $\s{R}$ of replacement channels is closed under tensor products, it follows from Lemma~\ref{lem:projective-mutual-information-choi} and the additivity of the max-relative entropy (and thus that of the Hilbert projective metric) that
\begin{align}
	I_\Omega(\ch{N}\otimes\ch{N}')&=\inf_{\ch{R}_{AA'\to BB'}^\sigma\in\s{R}}D_\Omega(\Phi^{\ch{N}\otimes\ch{N}'}\|\Phi^{\ch{R}^\sigma}) \notag\\
	&\leq\inf_{\ch{R}_{A\to B}^\sigma,\ch{R}_{A'\to B'}^{\sigma'}\in\s{R}}D_\Omega(\Phi^{\ch{N}\otimes\ch{N}'}\|\Phi^{\ch{R}^\sigma\otimes\ch{R}^{\sigma'}}) \notag\\
	&=\inf_{\ch{R}_{A\to B}^\sigma,\ch{R}_{A'\to B'}^{\sigma'}\in\s{R}}D_\Omega(\Phi^\ch{N}\otimes\Phi^{\ch{N}'}\|\Phi^{\ch{R}^\sigma}\otimes\Phi^{\ch{R}^{\sigma'}}) \notag\\
	&=\inf_{\ch{R}_{A\to B}^\sigma\in\s{R}}D_\Omega(\Phi^\ch{N}\|\Phi^{\ch{R}^\sigma})+\inf_{\ch{R}_{A'\to B'}^{\sigma'}\in\s{R}}D_\Omega(\Phi^{\ch{N}'}\|\Phi^{\ch{R}^{\sigma'}}) \notag\\
	&=I_\Omega(\ch{N})+I_\Omega(\ch{N}').
\end{align}

For the other direction, it follows from Eq.~\eqref{eq:projective-mutual-channel-dual} that
\begin{align}
	I_\Omega(\ch{N}\otimes\ch{N}')&=\log_2\sup_{P_{RR'BB'},Q_{RR'BB'}\geq0}\left\{\frac{\tr\!\left[P_{RR'BB'}\Phi_{RR'BB'}^{\ch{N}\otimes\ch{N}'}\right]}{\tr\!\left[Q_{RR'BB'}\Phi_{RR'BB'}^{\ch{N}\otimes\ch{N}'}\right]}\colon P_{BB'}=Q_{BB'}\right\} \notag\\
	&\geq\log_2\sup_{P_{RB},P_{R'B'}',Q_{RB},Q_{R'B'}'\geq0}\left\{\begin{array}[c]{c}
		\dfrac{\tr\!\left[\left(P_{RB}\otimes P_{R'B'}'\right)\Phi_{RR'BB'}^{\ch{N}\otimes\ch{N}'}\right]}{\tr\!\left[\left(Q_{RB}\otimes Q_{R'B'}'\right)\Phi_{RR'BB'}^{\ch{N}\otimes\ch{N}'}\right]}\colon \vspace{0.5em}\\
		P_B=Q_B,\;P_{B'}'=Q_{B'}'
	\end{array}\right\} \notag\\
	&=\log_2\sup_{P_{RB},Q_{RB}\geq0}\left\{\frac{\tr\!\left[P_{RB}\Phi_{RB}^\ch{N}\right]}{\tr\!\left[Q_{RB}\Phi_{RB}^\ch{N}\right]}\colon P_B=Q_B\right\} \notag\\
	&\qquad+\log_2\sup_{P_{R'B'}',Q_{R'B'}'\geq0}\left\{\frac{\tr\!\left[P_{R'B'}'\Phi_{R'B'}^{\ch{N}'}\right]}{\tr\!\left[Q_{R'B'}'\Phi_{R'B'}^{\ch{N}'}\right]}\colon P_{B'}'=Q_{B'}'\right\} \notag\\
	&=I_\Omega(\ch{N})+I_\Omega(\ch{N}').
\end{align}
\end{proof}

\begin{theorem}[Asymptotic pEA \& pNA capacities]
\label{thm:asymptotic}
Let $\ch{N}_{A\to B}$ be a quantum channel.  Then the asymptotic pEA and pNA capacities and their strong converses of $\ch{N}_{A\to B}$ are all captured by its projective mutual information:
\begin{align}
	I_\Omega(\ch{N})&=C_\sa{pEA}(\ch{N}) =C_\sa{pNA}(\ch{N})=2Q_\sa{pEA}(\ch{N})=2Q_\sa{pNA}(\ch{N}) \notag\\
	&=\stc{C}_\sa{pEA}(\ch{N})=\stc{C}_\sa{pNA}(\ch{N})=2\stc{Q}_\sa{pEA}(\ch{N})=2\stc{Q}_\sa{pNA}(\ch{N}).
\end{align}
\end{theorem}

\begin{proof}
By Theorem~\ref{thm:oneshot-classical}, and employing the property $\lceil\sqrt{r+1}-1\rceil\geq\frac{1}{2}\sqrt{r}$ for all $r\geq0$, we have that
\begin{align}
	C_\sa{pNA}^\varepsilon(\ch{N})\geq C_\sa{pEA}^\varepsilon(\ch{N})&\geq\log_2\!\left\lceil\sqrt{\frac{\varepsilon}{1-\varepsilon}2^{I_\Omega(\ch{N})}+1}-1\right\rceil^2 \notag\\
	&\geq\log_2\frac{\varepsilon}{4\left(1-\varepsilon\right)}+I_\Omega(\ch{N}).
\end{align}
Also by Theorem~\ref{thm:oneshot-classical}, and employing the fact $I_\Omega(\ch{N})\geq0$, we have that
\begin{align}
	C_\sa{pEA}^\varepsilon(\ch{N})\leq C_\sa{pNA}^\varepsilon(\ch{N})&\leq\log_2\!\left(\frac{\varepsilon}{1-\varepsilon}2^{I_\Omega(\ch{N})}+1\right) \notag\\
	&\leq\log_2\frac{1}{1-\varepsilon}+I_\Omega(\ch{N}).
\end{align}
Then Lemma~\ref{lem:additivity} implies that
\begin{align}
	\frac{1}{n}\log_2\frac{\varepsilon}{4\left(1-\varepsilon\right)}+I_\Omega(\ch{N})&\leq\frac{1}{n}C_\sa{pEA}^\varepsilon(\ch{N}^{\otimes n}) \notag\\
	&\leq\frac{1}{n}C_\sa{pNA}^\varepsilon(\ch{N}^{\otimes n})\leq\frac{1}{n}\log_2\frac{1}{1-\varepsilon}+I_\Omega(\ch{N}). \label{pf:asymptotic-1}
\end{align}
Since the limit as $n\to\infty$ of both sides of the equation exists and equals $I_\Omega(\ch{N})$  independently of $\varepsilon$, we have $C_\sa{pEA}(\ch{N})=\stc{C}_\sa{pEA}(\ch{N})=C_\sa{pNA}(\ch{N})=\stc{C}_\sa{pNA}(\ch{N})=I_\Omega(\ch{N})$.  Likewise, by Theorem~\ref{thm:oneshot-quantum}, we have that
\begin{align}
	\frac{1}{2}\log_2\frac{\varepsilon}{4\left(1-\varepsilon\right)}+\frac{1}{2}I_\Omega(\ch{N})&\leq Q_\sa{pEA}^\varepsilon(\ch{N}) \notag\\
	&\leq Q_\sa{pNA}^\varepsilon(\ch{N})\leq\frac{1}{2}\log_2\frac{1}{1-\varepsilon}+\frac{1}{2}I_\Omega(\ch{N}).
\end{align}
Then Lemma~\ref{lem:additivity} implies that
\begin{align}
	\frac{1}{2n}\log_2\frac{\varepsilon}{4\left(1-\varepsilon\right)}+\frac{1}{2}I_\Omega(\ch{N})&\leq\frac{1}{n}Q_\sa{pEA}^\varepsilon(\ch{N}^{\otimes n}) \notag\\
	&\leq\frac{1}{n}Q_\sa{pNA}^\varepsilon(\ch{N}^{\otimes n})\leq\frac{1}{2n}\log_2\frac{1}{1-\varepsilon}+\frac{1}{2}I_\Omega(\ch{N}), \label{pf:asymptotic-2}
\end{align}
and thus $Q_\sa{pEA}(\ch{N})=\stc{Q}_\sa{pEA}(\ch{N})=Q_\sa{pNA}(\ch{N})=\stc{Q}_\sa{pNA}(\ch{N})=\frac{1}{2}I_\Omega(\ch{N})$.
\end{proof}

\begin{remark}[No lower-order asymptotics for pEA \& pNA communication]
In the proof of Theorem~\ref{thm:asymptotic}, Eqs.~\eqref{pf:asymptotic-1} and \eqref{pf:asymptotic-2} imply that the optimal rates feature only the capacity term ($I_\Omega(\ch{N})$ or $\frac{1}{2}I_\Omega(\ch{N})$) and another $O\big(\frac{1}{n}\big)$ term, which echoes the similar findings of Ref.~\citen{regula_2022-4} for postselected hypothesis testing. This contrasts with the situation in  conventional entanglement-assisted communication, for which there is generally a nontrivial second-order $O\big(\frac{1}{\sqrt{n}}\big)$ term~\cite{Datta_2016}, and more generally so in other communication tasks (see Ref.~\citen{khatri_2021} and references therein).
\end{remark}

\begin{remark}[Feedback assistance provides no advantage in pEA \& pNA communication]
Feedback assistance refers to the extra resource of a noiseless channel from Bob (the receiver) to Alice (the sender) used to assist communication tasks from Alice to Bob~\cite{bowen_2004} (also see Ref.~\cite[Fig.~21.4]{wilde_2017} for an illustration).  In the multi-shot regime, such assistance enables `adaptive' protocols between different uses of the given channel $\ch{N}_{A\to B}$, and these protocols can be more general than parallel protocols (from which the asymptotic capacities are defined).  However, in pEA and pNA communication, we observe that every adaptive protocol based on feedback assistance can be simulated by a parallel protocol.  This is because Alice and Bob can always employ postselected teleportation from Bob to Alice to establish a postselected closed timelike curve, making use of shared entanglement and Bob's postselection, and such a curve simulates a noiseless feedback channel.  As a result, the pEA and pNA capacities remain the same even if feedback assistance is incorporated, and this holds not only asymptotically but also in the nonasymptotic regime.  This finding has some analogies in both the classical setting~\cite{shannon_1956} and the conventional quantum setting without postselection~\cite{bowen_2004,bowen_2005,ding_2015,cooney_2016}.
\end{remark}

\section{Conclusion}
\label{sec:conclusion}

In this paper, we introduced and studied the task of postselected communication over quantum channels.  In such a task, the receiver's decoding operation has an additional option of being inconclusive about the message being transmitted, and the pertaining error probability of message transmission is defined to be conditioned only on a conclusive outcome.  From an operational perspective, conditioning the error probability on a conclusive outcome amounts to allowing the receiver to perform postselection on his output system, a possibility not incorporated in the conventional theory of communication.  

After establishing the general framework and the concepts needed to study postselected communication, we completely characterised two specific scenarios.  The first scenario is postselected entanglement-assisted communication, wherein the sender and receiver are allowed to share an unbounded amount of entanglement to assist their communication over a given quantum channel.  The second scenario is postselected nonsignalling-assisted communication, wherein the sender and receiver's communication can be assisted by nonsignalling correlations.  In Theorem~\ref{thm:oneshot-quantum}, we derived the one-shot pEA and pNA quantum capacities of a given channel $\ch{N}_{A\to B}$, showing that both capacities are approximately characterised by an analytical expression in terms of the conditional error probability $\varepsilon$ and the channel's projective mutual information $I_\Omega(\ch{N})$.  In particular, we showed that these capacities can be achieved by a newly proposed pEA communication scheme (Protocol~\ref{prot:teleportation}) inspired by probabilistic teleportation.  We also established provably tight upper and lower bounds on the one-shot pEA and pNA classical capacities of $\ch{N}_{A\to B}$ in Theorem~\ref{thm:oneshot-classical}.  Informally, our results in the one-shot regime can be encapsulated as
\begin{align}
	C_\sa{pEA}^\varepsilon(\ch{N})&\approx C_\sa{pNA}^\varepsilon(\ch{N})\approx2Q_\sa{pEA}^\varepsilon(\ch{N})\approx2Q_\sa{pNA}^\varepsilon(\ch{N})\approx\log_2\!\left(\frac{\varepsilon}{1-\varepsilon}2^{I_\Omega(\ch{N})}+1\right).
\end{align}
In Theorem~\ref{thm:asymptotic}, we calculated the pEA and pNA capacities (both classical and quantum) and their strong converses in the asymptotic regime.  In either pEA or pNA communication, the quantum capacity and its strong converse of $\ch{N}_{A\to B}$ are both given by $\frac{1}{2}I_\Omega(\ch{N})$, and the classical capacity and its strong converse are twice as large, i.e., equal to $I_\Omega(\ch{N})$.  The simplicity of our results parallels that of Shannon's celebrated channel coding theorem for classical channels and the entanglement-assisted capacity theorem for quantum channels without postselection.

There are many other interesting problems to explore within our framework of postselected communication, and we leave them for future work.  First of all, it would be of interest to derive the postselected capacities of a channel in more restricted scenarios, e.g., the postselected randomness-assisted (pRA) scenario and the postselected unassisted (pU) scenario.  For pRA communication, our proof technique in Theorem~\ref{thm:oneshot-classical} can be suitably adapted to provide an upper bound on the one-shot pRA classical capacity,
\begin{align}
	C_\sa{pRA}^\varepsilon(\ch{N})&\leq\log_2\!\left(\frac{\varepsilon}{1-\varepsilon}2^{\chi_\Omega(\ch{N})}+1\right), \label{eq:pRA}
\end{align}
where $\chi_\Omega(\ch{N})$ is the projective variant of the Holevo quantity, defined as
\begin{align}
	\chi_\Omega(\ch{N})&:=\sup_{\rho_{XA}}I_\Omega(X;B)_{\ch{N}_{A\to B}\!\left[\rho_{XA}\right]},
\end{align}
with the supremum over all classical systems $X$ and states $\rho_{XA}$.  The postselected teleportation-based coding scheme (Protocol~\ref{prot:teleportation}) can also be modified to become a pRA classical protocol, by restricting the systems $M$, $A_1'$, $B_1'$, and $\g{M}$ to be classical.  However, we were not able to determine whether this modified protocol saturates the upper bound in Eq.~\eqref{eq:pRA}.

Going even further, one could also conceive a more general theory of postselected communication inspired by the one in this paper.  Specifically, in our framework, no penalty is imposed on the receiver if his decoding operation makes no conclusion on the message being transmitted; incorporating such a penalty would make the theory more attached to practical circumstances, albeit probably more complicated technically. The tradeoff between the error probability conditioned on a conclusive outcome and the conclusive probability would be an interesting topic to explore, and indeed variants of such questions have already been considered before~\cite{forney_1968,fiurasek_2003, rudolph_2003, croke_2006, herzog_2005, herzog_2009}.

\section*{Acknowledgements}
\label{sec:acknowledgements}

We are indebted to Ludovico Lami for several insightful discussions. We also thank Vincent Tan for helpful discussions about Refs.~\citen{forney_1968,merhav_2008,tan_2014,hayashi_2015}. KJ and MMW acknowledge support from the School of Electrical and Computer Engineering at Cornell University, and MMW further acknowledges support from the National Science Foundation under Grant No.~1907615.


\appendix

\section*{Appendices}

\renewcommand\thesection{\Alph{section}}%

\section{Equivalent formulations of the projective mutual information of a channel}
\label{app:equivalent}

\begin{lemma}[Choi formulation of $I_\Omega$]
\label{lem:projective-mutual-information-choi}
Let $\ch{N}_{A\to B}$ be a quantum channel.  Then
\begin{align}
	I_\Omega(\ch{N})&=\inf_{\ch{R}^\sigma_{A\to B}\in\s{R}}D_\Omega(\Phi^\ch{N}\|\Phi^{\ch{R}^\sigma}),
\end{align}
where the infimum is over all replacement channels $\ch{R}^\sigma_{A\to B}$.  This can be explicitly expressed as the optimal value of a semidefinite program:
\begin{align}\label{eq:sdp}
	I_\Omega(\ch{N})&=\log_2\inf_{\substack{\xi\in\spa{R}, \\ S_B\geq0}}\left\{\xi \colon\Phi^{\ch{N}}\leq\1_R\otimes S_B\leq\xi\Phi^{\ch{N}}\right\},
\end{align}
where the infimum exists whenever it is finite.
\end{lemma}

\begin{proof}
By definition, it holds that
\begin{align}
	I_\Omega(\ch{N})&=\sup_{\rho_{RA}} \inf_{\sigma_B}D_\Omega(\ch{N}_{A\to B}\!\left[\rho_{RA}\right]\|\rho_R\otimes\sigma_B).
\end{align}
But since $\rho_R\otimes\sigma_B=R^\sigma_{A\to B} (\rho_{RA})$ for every $\rho_{RA}$, we can write this as
\begin{align}
	I_\Omega(\ch{N})&=\sup_{\rho_{RA}}\inf_{\ch{R}^\sigma_{A\to B}\in\s{R}} D_\Omega(\ch{N}_{A\to B}\!\left[\rho_{RA}\right]\|\ch{R}^\sigma_{A\to B}\!\left[\rho_{RA}\right]).
\end{align}
Consider then that
\begin{align}
	\inf_{\ch{R}^\sigma_{A\to B}\in\s{R}}D_{\Omega}(\Phi^\ch{N}\|\Phi^{\ch{R}^\sigma})&\leq\sup_{\rho_{RA}}\inf_{\ch{R}^\sigma_{A\to B}\in\s{R}}D_\Omega(\ch{N}_{A\to B}\!\left[\rho_{RA}\right]\|\ch{R}^\sigma_{A\to B}\!\left[\rho_{RA}\right]) \notag\\
	&=\sup_{\rho_{RA}}\inf_{\ch{R}^\sigma_{A\to B}\in\s{R}}\left(D_{\max}(\ch{N}_{A\to B}\!\left[\rho_{RA}\right]\|\ch{R}^\sigma_{A\to B}\!\left[\rho_{RA}\right])\right. \notag\\
	&\qquad\left.+D_{\max}(\ch{R}^\sigma_{A\to B}\!\left[\rho_{RA}\right]\|\ch{N}_{A\to B}\!\left[\rho_{RA}\right])\right) \notag\\
	&\leq\inf_{\ch{R}^\sigma_{A\to B}\in\s{R}}\left(\sup_{\rho_{RA}'} D_{\max}\big(\ch{N}_{A\to B}\!\left[\rho_{RA}'\right]\|\ch{R}^\sigma_{A\to B}\!\left[\rho_{RA}'\right]\big)\right. \notag\\
	&\qquad\left.+\sup_{\rho_{RA}''} D_{\max}\big(\ch{R}^\sigma_{A\to B}\!\left[\rho_{RA}''\right]\|\ch{N}_{A\to B}\!\left[\rho_{RA}''\right]\big)\right) \notag\\
	&=\inf_{\ch{R}^\sigma_{A\to B}\in\s{R}}\left(D_{\max}(\Phi^\ch{N}\|\Phi^{\ch{R}^\sigma})+D_{\max}(\Phi^{\ch{R}^\sigma}\|\Phi^\ch{N})\right) \notag\\
	&=\inf_{\ch{R}^\sigma_{A\to B}\in\s{R}} D_{\Omega}(\Phi^\ch{N}\|\Phi^{\ch{R}^\sigma}),
\end{align}
where in the second-to-last line we used the fact that the supremum over input states in the definition of $D_{\max}$ between any two channels is always achieved on the maximally entangled state~\cite[Lemma~12]{wilde_2020}.

Now, using that
\begin{align}
	\left\{\Phi^{\ch{R}^\sigma}\colon\ch{R}^\sigma_{A\to B}\in\s{R}\right\}&=\left\{\frac{1}{d_R}\1_R\otimes\sigma_B\colon\sigma_B\geq0,\;\tr\!\left[\sigma\right]=1\right\},
\end{align}
from the definition of $D_{\max}$ [Eq.~\eqref{eq:max-relative}] we obtain
\begin{align}
	I_\Omega(\ch{N})&=\log_2\inf_{\substack{\lambda,\mu\in\spa{R}, \\ \sigma_B\geq0}}\left\{\lambda\mu\colon\Phi^\ch{N}\leq\frac{\lambda}{d_R}\1_R\otimes\sigma_B\leq\lambda\mu\Phi^\ch{N},\;\tr\!\left[\sigma_B\right]=1\right\} \notag\\
	&=\log_2\inf_{\substack{\lambda,\mu\in\spa{R}, \\ S_B\geq0}}\left\{\lambda\mu\colon\Phi^\ch{N}\leq\1_R \otimes S_B\leq\lambda\mu\Phi^\ch{N}\right\},
\end{align}
where in the second line we noticed that any unnormalised $S_B$ can be rescaled as $\sigma_B=S_B/\tr[S_B]$ yielding a feasible state with $\lambda'=d_R\tr[S_B]$ and $\mu'=\lambda\mu/(d_R\tr[S_B])$, and this does not affect the optimal value since $\lambda'\mu'=\lambda\mu$.  The closedness of the positive semidefinite cone ensures that, whenever the program has any feasible solution and is therefore bounded, an optimal solution must exist (see Ref.~\cite[Theorem~1(iii)]{regula_2022-2}).
\end{proof}

\begin{lemma}[Dual formulation of $I_\Omega$]
\label{lem:projective-mutual-information-dual}
Let $\ch{N}_{A\to B}$ be a quantum channel.  Then
\begin{align}
	I_\Omega(\ch{N})&=\log_2\sup_{P_{RB},Q_{RB}\geq0}\left\{\frac{\tr\!\left[P_{RB}\Phi_{RB}^\ch{N}\right]}{\tr\!\left[Q_{RB}\Phi_{RB}^\ch{N}\right]}\colon P_B=Q_B\right\},
\end{align}
where the supremum is over all positive semidefinite operators $P_{RB}$ and $Q_{RB}$.
\end{lemma}

\begin{proof}
Taking the dual of the semidefinite program in Eq.~\eqref{eq:sdp} gives the optimisation problem
\begin{align}
	\dual{I}_\Omega(\ch{N})&=\log_2\sup_{P_{RB},Q_{RB}\geq0}\left\{\tr\!\left[P_{RB}\Phi_{RB}^\ch{N}\right]\colon\tr\!\left[Q_{RB}\Phi_{RB}^\ch{N}\right]=1,\;P_B\leq Q_B\right\}.
\end{align}
Since any feasible solution $(P_{RB},Q_{RB})$ can be rescaled as $Q_{RB}\mapsto Q_{RB}/\tr[Q_{RB}\Phi_{RB}^\ch{N}]$ and $P_{RB}\mapsto P_{RB}/\tr[Q_{RB}\Phi_{RB}^\ch{N}]$, without affecting the optimal value, we can write the above as
\begin{align}
	\dual{I}_\Omega(\ch{N})&=\log_2\sup_{P_{RB},Q_{RB}\geq0}\left\{\frac{\tr\!\left[P_{RB}\Phi_{RB}^\ch{N}\right]}{\tr\!\left[Q_{RB}\Phi_{RB}^\ch{N}\right]}\colon P_B\leq Q_B\right\}.
\end{align}
Now, for every feasible $(P_{RB},Q_{RB})$, we can define a feasible solution $(P_{RB}',Q_{RB}')$ with $P_{RB}':=P_{RB}+\1_R\otimes(Q_B-P_B)/d_R$ and $Q_{RB}':=Q_{RB}$ whose optimal value cannot be smaller than the original one. Noticing that $P_B'=Q'_B$, we can thus write
\begin{align}
	\dual{I}_\Omega(\ch{N})&=\log_2\sup_{P_{RB},Q_{RB}\geq0}\left\{\frac{\tr\!\left[P_{RB}\Phi_{RB}^\ch{N}\right]}{\tr\!\left[Q_{RB}\Phi_{RB}^\ch{N}\right]}\colon P_B= Q_B\right\}.
\end{align}
Since $P_{RB}=Q_{RB}=\1_{RB}$ is strictly feasible for the above, strong duality holds and we have $I_\Omega(\ch{N})=\dual{I}_\Omega(\ch{N})$.
\end{proof}

\section{Proof of Proposition~\ref{prop:pNA}}
\label{app:pNA}

{
\newcounter{tempprop}
\setcounter{tempprop}{\value{proposition}}
\setcounter{proposition}{40}
\renewcommand{\theproposition}{\ref{prop:pNA}}
\begin{proposition}[Restatement]
Let $\Theta_{(A\to B)\to(M\to\g{M})}$ be a probabilistic supermap, which can be thought of as a bipartite subchannel $\bc{\Theta}_{MB\to A\g{M}}$ between Alice and Bob (see Fig.~\ref{fig:bipartite}).  Then the following statements are equivalent.
\begin{itemize}
	\item[\textnormal{(i)}] There exists a pNA protocol $\Lambda_{(A\to B)\to(M\to\g{M})}$ and $c\geq0$ such that 
    \begin{align}
        \Theta_{(A\to B)\to(M\to\g{M})}&=c\Lambda_{(A\to B)\to(M\to\g{M})}.
    \end{align}
	\item[\textnormal{(ii)}] $\bc{\Theta}_{MB\to A\g{M}}$ is Alice-to-Bob nonsignalling in the sense of Eq.~\eqref{eq:NA} (see Fig.~\ref{fig:pNA}).
	\item[\textnormal{(iii)}] For every replacement channel $\ch{R}_{A\to B}^\sigma$, there exists a state $\sigma_\g{M}'$ and $p\in[0,1]$ such that
    \begin{align}
        \Theta_{(A\to B)\to(M\to\g{M})}\left\{\ch{R}_{A\to B}^\sigma\right\}&=p\ch{R}_{M\to\g{M}}^{\sigma'}.
    \end{align}
\end{itemize}
\end{proposition}
\setcounter{proposition}{\value{tempprop}}
}

\begin{proof}[Proof of (i) $\Rightarrow$ (iii)]
Consider a probabilistic supermap $\Theta_{(A\to B)\to(M\to\g{M})}$ satisfying (i).  By Definition~\ref{def:pNA}, $\Theta_{(A\to B)\to(M\to\g{M})}$ can be decomposed into an NA protocol $\Xi_{(A\to B)\to(M\to\g{M}X)}$, a subchannel $\ch{D}_{\g{M}X\to\g{M}}$, and a factor $c\in\spa{R}$.  Since the NA protocol $\Xi_{(A\to B)\to(M\to\g{M}X)}$ is replacement preserving~\cite[Proposition~1]{takagi_2020}, there exists a state $\tau_{\g{M}X}$ such that 
\begin{align}
	\ch{R}_{M\to\g{M}X}^\tau&=\Xi_{(A\to B)\to(M\to\g{M}X)}\left\{\ch{R}_{A\to B}^\sigma\right\}.
\end{align}
Define $p:=c\tr[\ch{D}_{\g{M}X\to\g{M}}[\tau_{\g{M}X}]]$ and a state $\sigma_\g{M}':=c\ch{D}_{\g{M}X\to\g{M}}[\tau_{\g{M}X}]/p$.  (If $p=0$, then it suffices to let $\sigma_\g{M}'$ be an arbitrary state.)  Then we have that
\begin{align}
	\Theta_{(A\to B)\to(M\to\g{M})}\left\{\ch{R}_{A\to B}^\sigma\right\}\!\left[\rho_M\right]&=c\ch{D}_{\g{M}X\to\g{M}}\circ\Xi_{(A\to B)\to(M\to\g{M}X)}\left\{\ch{R}_{A\to B}^\sigma\right\}\!\left[\rho_M\right] \notag\\
	&=c\ch{D}_{\g{M}X\to\g{M}}\circ\ch{R}_{M\to\g{M}X}^\tau\!\left[\rho_M\right] \notag\\
	&=c\tr\!\left[\rho_M\right]\ch{D}_{\g{M}X\to\g{M}}\!\left[\tau_{\g{M}X}\right] \notag\\
	&=p\tr\!\left[\rho_M\right]\sigma_\g{M}' \notag\\
	&=p\ch{R}_{M\to\g{M}}^{\sigma'}\!\left[\rho_M\right]\qquad\forall\rho_M.
\end{align}
This shows $\Theta_{(A\to B)\to(M\to\g{M})}\{\ch{R}_{A\to B}^\sigma\}=p\ch{R}_{M\to\g{M}}^{\sigma'}$.  Since $\Theta_{(A\to B)\to(M\to\g{M})}\{\ch{R}_{A\to B}^\sigma\}$ is trace nonincreasing, we must have $p\in[0,1]$.
\end{proof}

\begin{proof}[Proof of (iii) $\Rightarrow$ (ii)]
Consider a probabilistic supermap satisfying (iii).  Denote $\ch{M}_{M\to\g{M}}^\sigma\equiv\Theta_{(A\to B)\to(M\to\g{M})}\{\ch{R}_{A\to B}^\sigma\}$.  For every state $\sigma_B$, there exists a state $\sigma_\g{M}'$ and $p\in[0,1]$ such that $\ch{M}_{M\to\g{M}}^\sigma=p\ch{R}_{M\to\g{M}}^{\sigma'}$.  This implies that
\begin{align}
	\ch{M}_{M\to\g{M}}^\sigma\!\left[\rho_M\right]&=\ch{M}_{M\to\g{M}}^\sigma\!\left[\omega_M\right] \label{pf:pNA-1}
\end{align}
for all states $\rho_M$ and $\omega_M$.  The correspondence between $\Theta_{(A\to B)\to(M\to\g{M})}$ and its equivalent bipartite subchannel $\bc{\Theta}_{MB\to A\g{M}}$, translates Eq.~\eqref{pf:pNA-1} to Eq.~\eqref{eq:NA} straightforwardly.
\end{proof}

\begin{proof}[Proof of (ii) $\Rightarrow$ (i)]
Consider a probabilistic supermap $\Theta_{(A\to B)\to(M\to\g{M})}$ satisfying (ii).  Let $\Upsilon_{(A\to B)\to(M\to\g{M})}$ denote the completely depolarising superchannel, which is defined by
\begin{align}
	\Upsilon_{(A\to B)\to(M\to\g{M})}\left\{\ch{N}_{A\to B}\right\}&:=\ch{C}_{B\to\g{M}}\circ\ch{N}_{A\to B}\circ\ch{C}_{M\to A}\qquad\forall\ch{N}_{A\to B},
\end{align}
where $\ch{C}_{B\to\g{M}}:=\ch{R}_{B\to\g{M}}^{\1/d_\g{M}}$ denotes the completely depolarising channel and satisfies $\Phi_{R\g{M}}^\ch{C}=\1^{R\g{M}}/d_Bd_\g{M}$ (and similarly for $\ch{C}_{M\to A}$).  Define
\begin{align}
	\Lambda_{(A\to B)\to(M\to\g{M})}&:=\Upsilon_{(A\to B)\to(M\to\g{M})}-\frac{1}{d_M^2d_Ad_B}\Theta_{(A\to B)\to(M\to\g{M})}.
\end{align}
Our next step is to show that $\Lambda_{(A\to B)\to(M\to\g{M})}$ is a probabilistic supermap.  Utilising the necessary and sufficient condition provided in Ref.~\citen{burniston_2020}, this amounts to showing that $\Phi_{MBA\g{M}}^{\bc{\Lambda}}\geq0$ and that $\Lambda_{(A\to B)\to(M\to\g{M})}$ satisfies Ref.~\cite[Eq.~(4.2)]{burniston_2020}, where $\bc{\Lambda}_{MB\to A\g{M}}$ is the equivalent bipartite subchannel of $\Lambda_{(A\to B)\to(M\to\g{M})}$.

Due to the fact that $\bc{\Upsilon}_{MB\to A\g{M}}$ of $\Upsilon_{(A\to B)\to(M\to\g{M})}$ is a completely depolarising channel, we know that $\bc{\Lambda}_{MB\to A\g{M}}=\bc{\Upsilon}_{MB\to A\g{M}}-\bc{\Theta}_{MB\to A\g{M}}/(d_M^2d_Ad_B)$ is a subchannel, and this shows $\Phi_{MBA\g{M}}^{\bc{\Lambda}}\geq0$.  On the other hand, since the superchannel $\Upsilon_{(A\to B)\to(M\to\g{M})}$ satisfies Ref.~\cite[Eq.~(4.2)]{burniston_2020}, and since the Choi operator of the probabilistic subchannel $\Theta_{(A\to B)\to(M\to\g{M})}$ is positive semidefinite, we know that $\Lambda_{(A\to B)\to(M\to\g{M})}$ also satisfies Ref.~\cite[Eq.~(4.2)]{burniston_2020}.  This shows that $\Lambda_{(A\to B)\to(M\to\g{M})}$ is probabilistic supermap. 

Knowing that $\Lambda_{(A\to B)\to(M\to\g{M})}$ is a probabilistic supermap, and recalling that $\Upsilon_{(A\to B)\to(M\to\g{M})}=\Theta_{(A\to B)\to(M\to\g{M})}+\Lambda_{(A\to B)\to(M\to\g{M})}$ is a superchannel, it can be verified that the following is a superchannel (see Ref.~\cite[Theorem~1]{gour_2019-2}):
\begin{align}
	\Xi_{(A\to B)\to(M\to\g{M}X)}&:=\Lambda_{(A\to B)\to(M\to\g{M})}\otimes\op{0}{0}_X \notag\\
    &\qquad+\frac{1}{d_M^2d_Ad_B}\Theta_{(A\to B)\to(M\to\g{M})}\otimes\op{1}{1}_X. \label{pf:pNA-2}
\end{align}
Note that $\Theta_{(A\to B)\to(M\to\g{M})}/(d_M^2d_Ad_B)$ can be realised by composing $\Xi_{(A\to B)\to(M\to\g{M}X)}$ with the following subchannel:
\begin{align}
	\ch{D}_{\g{M}X\to\g{M}}\!\left[\rho_{\g{M}X}\right]&:=\bra{1}_X\rho_{\g{M}X}\ket{1}_X\qquad\forall\rho_{\g{M}X}.
\end{align}
Consequently, to arrive at (i), it suffices to show that $\Xi_{(A\to B)\to(M\to\g{M}X)}$ is an NA protocol.

To show this, note that Eq.~\eqref{eq:NA} is an affine constraint.  This means that $\Lambda_{(A\to B)\to(M\to\g{M})}$ satisfies (ii) since both $\Upsilon_{(A\to B)\to(M\to\g{M})}$ and $\Theta_{(A\to B)\to(M\to\g{M})}$ do.  By Eq.~\eqref{pf:pNA-2}, the superchannel $\Xi_{(A\to B)\to(M\to\g{M}X)}$ satisfies (ii) as well, and therefore by definition is an NA protocol (see Ref.~\cite[Eqs.~(5) and (6)]{takagi_2020}). 
\end{proof}

\section{Alternative proof of the lower bound on the one-shot pNA quantum capacity}
\label{app:achievability}

\begin{proposition}[Achievability in one-shot pNA communication]\label{prop:achiev}
Let $\ch{N}_{A\to B}$ be a quantum channel, and let $\varepsilon\in(0,1)$.  Then, for every
\begin{align}
	d_M&< \sqrt{\frac{\varepsilon}{1-\varepsilon}2^{I_\Omega(\ch{N})}+1},
\end{align}
there exists a pNA $(d_M,\varepsilon)$ quantum protocol over $\ch{N}_{A\to B}$.
\end{proposition}

\begin{proof}
\let\ve\varepsilon
The proof idea follows Ref.~\cite[Theorems~5 and 12]{regula_2022-2}. We will show that, as long as
\begin{align}
	I_\Omega(\ch{N})&>\log_2\!\left(\frac{1-\ve}{\ve}\left(d_M^2-1\right)\right),
\end{align} 
then there exists a probabilistically replacement-preserving protocol $\Theta_{(A\to B)\to(M\to\g{M})}$ which transforms $\ch{N}$ to a channel that is $\ve$-close to the target identity channel $\id_{M\to\g{M}}$. By Proposition~\ref{prop:pNA} and Remark~\ref{rem:rescaling}, this will yield a valid pNA protocol after proper rescaling.

Let $\ch{C}_{M\to\g{M}}$ denote the completely depolarising channel, which satisfies $\Phi_{R\g{M}}^{\ch{C}}=\1_{R\g{M}}/d_M^2$ and $\ch{C}_{M\to\g{M}}\in\s{R}$. We will use the easily verifiable fact that $\id_{M\to\g{M}}\leq d_M^2\ch{C}_{M\to\g{M}}$, where the inequality is understood as an operator inequality between the corresponding Choi states.

Let us first consider the case of a very large error, $\ve>\frac{d_M^2-1}{d_M^2}$, as then the transformation trivialises. Namely, here we have that
\begin{align}
	\ch{C}_{M\to\g{M}}&=\left(1-\ve\right)d_M^2\ch{C}_{M\to\g{M}}+\left(\ve-\left(1-\ve\right)\left(d_M^2-1\right)\right)\ch{C}_{M\to\g{M}} \notag\\
	&=\left(1-\ve\right)\id_{M\to\g{M}}+\left(1-\ve\right)\left(d_M^2\ch{C}_{M\to\g{M}}-\id_{M\to\g{M}}\right) \notag\\
	&\qquad+\left(\ve-\left(1-\ve\right)\left(d_M^2-1\right)\right)\ch{C}_{M\to\g{M}},
\end{align}
which already satisfies
\begin{align}
	\inf_{\psi_{RM}}\tr\!\left[\psi_{R\g{M}}\ch{C}_{M\to\g{M}}\!\left[\psi_{RM}\right]\right]&\geq1-\ve.
\end{align}
One can then simply discard the channel $\ch{N}_{A\to B}$ and prepare $\ch{C}_{M\to\g{M}}$ instead to achieve the desired error.

Assume now that $\ve\leq\frac{d_M^2-1}{d_M^2}$. Since $I_\Omega(\ch{N})>\log_2(\frac{1-\ve}{\ve}(d_M^2-1))$, there exists an $\ve'<\ve$ such that
\begin{align}
	I_\Omega(\ch{N})&\geq\log_2\!\left(\frac{1-\ve'}{\ve'}\left(d_M^2-1\right)\right).
\end{align}
We will show that $\ch{N}_{A\to B}$ can be probabilistically transformed into a channel arbitrarily close to
\begin{align}
	\ch{T}_{M\to\g{M}}^{\ve'}&:=q\left(1-\ve'\right)\id_{M\to\g{M}}+\frac{\ve'}{d_M^2-1}\left(d_M^2\ch{C}_{M\to\g{M}}-\id_{M\to\g{M}}\right),
\end{align}
which clearly has
\begin{align}
	\inf_{\psi_{RM}}\tr\!\left[\psi_{R\g{M}}\ch{T}_{M\to\g{M}}^{\ve'}\!\left[\psi_{RM}\right]\right]&\geq1-\ve'
\end{align}
by construction. The fact that $\ve'\leq\frac{d_M^2-1}{d_M^2}$ implies that $(1-\ve')-\frac{\ve'}{d_M^2-1}\geq0$, and hence
\begin{align}
	\ch{T}_{M\to\g{M}}^{\ve'}&=\left(\left(1-\ve'\right)-\frac{\ve'}{d_M^2-1}\right)\id_{M\to\g{M}}+\frac{\ve'd_M^2}{d_M^2-1}\ch{C}_{M\to\g{M}} \notag\\
	&\leq\left(\left(1-\ve'\right)-\frac{\ve'}{d_M^2-1}\right)d_M^2\ch{C}_{M\to\g{M}}+\frac{\ve'd_M^2}{d_M^2-1}\ch{C}_{M\to\g{M}} \notag\\
	&=\left(1-\ve'\right)d_M^2\ch{C}_{M\to\g{M}}
\end{align}
as well as
\begin{align}
	\ch{T}_{M\to\g{M}}^{\ve'}&\geq\frac{\ve'd_M^2}{d_M^2-1}\ch{C}_{M\to\g{M}}.
\end{align}
From the Choi form of $I_\Omega$ in Lemma~\ref{lem:projective-mutual-information-choi}, it follows that
\begin{align}\label{eq:decreases}
	I_\Omega(\ch{T}^{\ve'})&\leq\log_2\!\left(\frac{1-\ve'}{\ve'}\left(d_M^2-1\right)\right)\leq I_\Omega(\ch{N}).
\end{align}
Let $\lambda$ and $\mu$ be such that $\log_2(\lambda\mu)=I_\Omega(\ch{N})$; i.e., there exists a replacement channel $\ch{R}_{A\to B}^\sigma$ such that $\ch{N}_{A\to B}\leq\lambda\ch{R}_{A\to B}^\sigma$ and $\ch{R}_{A\to B}^\sigma\leq\mu\ch{N}_{A\to B}$. Due to Eq.~\eqref{eq:decreases}, up to rescaling by some constant $r>0$, it holds that $r\ch{T}_{M\to\g{M}}^{\ve'}\leq\lambda\ch{C}_{M\to\g{M}}$ and $\ch{C}_{M\to\g{M}}\leq\mu r\ch{T}_{M\to\g{M}}^{\ve'}$. 

Let $(P_{RB},Q_{RB})$ be an arbitrary feasible solution for $I_\Omega(\ch{N})$ in Lemma~\ref{lem:projective-mutual-information-dual}. Define the protocol $\Theta_{(A\to B)\to(M\to\g{M})}$ as
\begin{align}
	\Theta_{(A\to B)\to(M\to\g{M})}\left\{\ch{N}_{A\to B}\right\}&:=\tr\!\left[Q_{RB}\Phi_{RB}^\ch{N}\right]\left(\lambda\ch{C}_{M\to\g{M}}-r\ch{T}_{M\to\g{M}}^{\ve'}\right) \notag\\
	&\qquad+\tr\!\left[P_{RB}\Phi_{RB}^\ch{N}\right]\left(r\ch{T}_{M\to\g{M}}^{\ve'}-\frac{1}{\mu}\ch{C}_{M\to\g{M}}\right) \notag\\
	&\hphantom{:}=\tr\!\left[Q_{RB}\Phi_{RB}^\ch{N}\right]\left(\left(\lambda\ch{C}_{M\to\g{M}}-r\ch{T}_{M\to\g{M}}^{\ve'}\right)\vphantom{+\frac{\tr\!\left[P_{RB}\Phi_{RB}^\ch{N}\right]}{\tr\!\left[Q_{RB}\Phi_{RB}^\ch{N}\right]}\left(r\ch{T}_{M\to\g{M}}^{\ve'}-\frac{1}{\mu}\ch{C}_{M\to\g{M}}\right)}\right. \notag\\
	&\qquad\left.+\frac{\tr\!\left[P_{RB}\Phi_{RB}^\ch{N}\right]}{\tr\!\left[Q_{RB}\Phi_{RB}^\ch{N}\right]}\left(r\ch{T}_{M\to\g{M}}^{\ve'}-\frac{1}{\mu}\ch{C}_{M\to\g{M}}\right)\right). \label{pf:achievability-pNA}
\end{align}
Up to normalisation, $\Theta_{(A\to B)\to(M\to\g{M})}$ is clearly a valid probabilistic supermap  --- we first prepare the maximally entangled state state $\Phi_{RA}$, apply $\ch{N}_{A\to B}$ to one half of it in order to get $\Phi_{RB}^\ch{N}$, then perform the incomplete measurement $\{Q_{RB},P_{RB}\}$, which represents the conclusive outcomes of some POVM (it can be rescaled and completed to a valid measurement). Finally, the protocol prepares a subchannel depending on the measurement outcome. 

To verify that $\Theta_{(A\to B)\to(M\to\g{M})}$ is probabilistically replacement preserving, we check that for every $\ch{R}^\sigma\in\s{R}$, it holds that
\begin{align}
	\Theta_{(A\to B)\to(M\to\g{M})}\left\{\ch{R}_{A\to B}^\sigma\right\}&=\tr\!\left[Q_{RB}\Phi_{RB}^{\ch{R}^\sigma}\right]\left(\lambda\ch{C}_{M\to\g{M}}-\frac{1}{\mu}\ch{C}_{M\to\g{M}}\right) \notag\\
	&\propto\ch{C}_{M\to\g{M}}\in\s{R}
\end{align}
due to the fact that $\tr[\Phi_{RB}^{\ch{R}^\sigma}P_{RB}]=\tr[\Phi_{RB}^{\ch{R}^\sigma}Q_{RB}]$ for all replacement channels $\ch{R}^\sigma$.  Thus, $\Theta_{(A\to B)\to(M\to\g{M})}$ is probabilistically replacement preserving.   By Proposition~\ref{prop:pNA}, $\Theta_{(A\to B)\to(M\to\g{M})}$ is a multiple of a pNA protocol, and thus it can be rescaled to be a pNA protocol without altering its performance in postselected communication (see Remark~\ref{rem:rescaling}).

Crucially, following Eq.~\eqref{pf:achievability-pNA}, the simulated subchannel $\ch{N}_{M\to\g{M}}'\equiv\Theta_{(A\to B)\to(M\to\g{M})}\{\ch{N}_{A\to B}\}$ satisfies
\begin{align}
	\ch{N}_{M\to\g{M}}'&=\tr\!\left[Q_{RB}\Phi_{RB}^\ch{N}\right]\left(\frac{1}{\mu}\left(\lambda\mu-\frac{\tr\!\left[P_{RB}\Phi_{RB}^\ch{N}\right]}{\tr\!\left[Q_{RB}\Phi_{RB}^\ch{N}\right]}\right)\ch{C}_{M\to\g{M}}\right. \notag\\
	&\qquad\left.+r\left(\frac{\tr\!\left[P_{RB}\Phi_{RB}^\ch{N}\right]}{\tr\!\left[Q_{RB}\Phi_{RB}^\ch{N}\right]}-1\right)\ch{T}_{M\to\g{M}}^{\ve'}\right) \notag\\
	&=\tr\!\left[Q_{RB}\Phi_{RB}^\ch{N}\right]\left(\frac{1}{\mu}\left(2^{I_\Omega(\ch{N})}-\frac{\tr\!\left[P_{RB}\Phi_{RB}^\ch{N}\right]}{\tr\!\left[Q_{RB}\Phi_{RB}^\ch{N}\right]}\right)\ch{C}_{M\to\g{M}}\right. \notag\\
	&\qquad\left.+r\left(\frac{\tr\!\left[P_{RB}\Phi_{RB}^\ch{N}\right]}{\tr\!\left[Q_{RB}\Phi_{RB}^\ch{N}\right]}-1\right)\ch{T}_{M\to\g{M}}^{\ve'}\right).
\end{align}
Due to the dual formulation of $I_\Omega(\ch{N})$ in Lemma~\ref{lem:projective-mutual-information-dual}, $P_{RB}$ and $Q_{RB}$ can be chosen so that $\tr\!\left[P_{RB}\Phi_{RB}^\ch{N}\right]/\tr\!\left[Q_{RB}\Phi_{RB}^\ch{N}\right]$ is arbitrarily close to $2^{I_\Omega(\ch{N})}$, so $\ch{N}_{M\to\g{M}}'$ can be made arbitrarily close to a multiple of $\ch{T}_{M\to\g{M}}^{\ve'}$. Specifically, let us choose $P_{RB}$ and $Q_{RB}$ so that $2^{I_\Omega(\ch{N})}-\tr\!\left[P_{RB}\Phi_{RB}^\ch{N}\right]/\tr\!\left[Q_{RB}\Phi_{RB}^\ch{N}\right]=\delta$ for some $\delta>0$. For every pure state $\psi_{RM}$, we then have that
\begin{align}
	\frac{\tr\!\left[\psi_{R\g{M}}\ch{N}_{M\to\g{M}}'\!\left[\psi_{RM}\right]\right]}{\tr\!\left[\ch{N}_{M\to\g{M}}'\!\left[\psi_{RM}\right]\right]}&\geq\frac{r\left(2^{I_\Omega(\ch{N})}-\delta-1\right)}{\frac{\delta}{\mu}+r\left(2^{I_\Omega(\ch{N})}-\delta-1\right)}\tr\!\left[\psi_{R\g{M}}\ch{T}_{M\to\g{M}}^{\ve'}\!\left[\psi_{RM}\right]\right] \notag\\
	&\geq\frac{r\left(2^{I_\Omega(\ch{N})}-\delta-1\right)}{\frac{\delta}{\mu}+r\left(2^{I_\Omega(\ch{N})}-\delta-1\right)}\left(1-\ve'\right).
\end{align}
Choosing $\delta$ small enough, we can thus guarantee that
\begin{align}
	\frac{\tr\!\left[\psi_{R\g{M}}\ch{N}_{M\to\g{M}}'\!\left[\psi_{RM}\right]\right]}{\tr\!\left[\ch{N}_{M\to\g{M}}'\!\left[\psi_{RM}\right]\right]}&\geq 1-\ve,
\end{align}
as was to be shown.
\end{proof}

\begin{remark}[On the use of a strict inequality in Proposition~\ref{prop:achiev}]
What necessitates the use of a strict inequality in the statement of Proposition~\ref{prop:achiev}, thus preventing a tighter result, is the fact that the dual formulation of $I_\Omega$ (Lemma~\ref{lem:projective-mutual-information-dual}) is not necessarily achieved. If, for a particular $\ch{N}_{A\to B}$, there exist optimal $P_{RB}$ and $Q_{RB}$, then the above can be tightened to show that there exists a pNA protocol with
\begin{align}
	d_M&=\left\lfloor\sqrt{\frac{\varepsilon}{1-\varepsilon}2^{I_\Omega(\ch{N})}+1}\right\rfloor.
\end{align}
\end{remark}


\bibliographystyle{unsrt}
\bibliography{main}

\end{document}